\documentclass[lettersize,journal]{IEEEtran}
\IEEEoverridecommandlockouts

\ifCLASSINFOpdf
\else
\fi
\usepackage{color}
\usepackage{lettrine}
\usepackage{citesort}
\usepackage{cite}
\usepackage{amsfonts}
\usepackage{amssymb}
\usepackage{amsmath}
\usepackage{mathrsfs}
\usepackage[dvips]{graphicx}
\usepackage{verbatim}
\usepackage{subfigure}
\usepackage{balance}
\usepackage{booktabs}
\usepackage{fancybox}
\usepackage{bm}
\usepackage{extarrows}
\usepackage{algorithm}
\usepackage{algorithmic}
\usepackage{multirow}
\usepackage{array}
\newtheorem{theorem}{Theorem}

\newtheorem{corollary}{Corollary}
\newtheorem{proof}{Proof}

\usepackage{lipsum}
\usepackage{graphicx}

\def\BibTeX{{\rm B\kern-.05em{\sc i\kern-.025em b}\kern-.08em
    T\kern-.1667em\lower.7ex\hbox{E}\kern-.125emX}}
\begin{document}

\title{Channel Reciprocity Attacks Using Intelligent Surfaces with Non-Diagonal Phase Shifts}
\author
{\IEEEauthorblockN{Haoyu~Wang, Zhu Han~\IEEEmembership{Fellow, IEEE}, and A.~Lee Swindlehurst,~\IEEEmembership{Fellow, IEEE}}
\thanks{H. Wang and A.~L.~Swindlehurst are with the Center for Pervasive Communications and Computing, University of California at Irvine, Irvine, CA 92697 USA (e-mail: haoyuw30@uci.edu, swindle@uci.edu). Z. Han is with the Department of Electrical and Computer Engineering at the University of Houston, Houston, TX 77004 USA, and also with the Department of Computer Science and Engineering, Kyung Hee University, Seoul, South Korea, 446-701 (email: hanzhu22@gmail.com).}
%\thanks{This work was supported by the U.S.~National Science Foundation under grants ECCS-2030029 and CNS-2107182.}
}

\maketitle

\newcommand\blfootnote[1]{%
\begingroup
\renewcommand\thefootnote{}\footnote{#1}%
\addtocounter{footnote}{-1}%
\endgroup
}
% The paper headers

%\markboth{Journal of \LaTeX\ Class Files,~Vol.~11, No.~4, December~2012}%
%{Shell \MakeLowercase{\textit{et al.}}: Bare Demo of IEEEtran.cls for Journals}

%\begin{comment}

\begin{abstract}
While reconfigurable intelligent surface (RIS) technology has been shown to provide numerous benefits to wireless systems, in the hands of an adversary such technology can also be used to disrupt communication links. This paper describes and analyzes an RIS-based attack on multi-antenna wireless systems that operate in time-division duplex mode under the assumption of channel reciprocity. In particular, we show how an RIS with a non-diagonal (ND) phase shift matrix (referred to here as an ND-RIS) can be deployed to maliciously break the channel reciprocity and hence degrade the downlink network performance. Such an attack is entirely passive and difficult to detect and counteract. We provide a theoretical analysis of the degradation in the sum ergodic rate that results when an arbitrary malicious ND-RIS is deployed and design an approach based on the genetic algorithm for optimizing the ND structure under partial knowledge of the available channel state information. Our simulation results validate the analysis and demonstrate that an ND-RIS channel reciprocity attack can dramatically reduce the downlink throughput.
\end{abstract}

\begin{IEEEkeywords}
Channel Reciprocity, Passive Attack, Channel Estimation, Reconfigurable Intelligent Surfaces, Physical Layer Security.
\end{IEEEkeywords}

\IEEEpeerreviewmaketitle

\begin{comment}
\renewcommand{\thefootnote}{}
\footnotetext{The work of H. Wang, Y. Ju and Q. Pei was supported in part by the National Key Research and Development Program of China under Grant 2020YFB1807500 and Grant 2018YFE0126000, and in part by the Guangdong Basic and Applied Basic Research Foundation under Grant 2020A1515110772. The work of H.-M. Wang was supported in part by the National Natural Science Foundation of China under Grant 61941105 and Grant 61941118, in part by the Outstanding Young Research Fund of Shaanxi Province under Grant 2018JC-003, and in part by the Innovation Team Research Fund of Shaanxi Province under Grant 2019TD-013. \textit{(Corresponding author: Ying Ju, e-mail: juyingtju@163.com.)}}
\end{comment}

\blfootnote{The work of H. Wang and A. L. Swindlehurst were supported by the U.S. National Science Foundation (NSF) under grants ECCS-2030029, and CNS-2107182. The work of Z.~Han was supported by the US Department of Transportation, Amazon, Toyota, and NSF grants CNS-2107216, CNS-2128368, CMMI-2222810, and ECCS-2302469.}

\section{ INTRODUCTION}\label{s1}
Wireless communication has become an indispensable component of everyday life and is poised to offer a plethora of reliable data services in the future. As smart devices continue to proliferate at a rapid pace, ensuring the security of wireless communication has emerged as a paramount concern~\cite{7467419}. Owing to the inherent traits of the open communications environment, encompassing the superposition and broadcast nature of wireless channels, wireless networks are susceptible to malicious attacks or eavesdropping. Consequently, it becomes crucial and yet challenging to guarantee the secure transmission of wireless data. Encryption is a commonly used tool for improving wireless security, but it comes with high computational costs and complexity. In recent years, physical layer techniques involving jamming or spatial selectivity with multiple antennas have been proposed as alternatives~\cite{6739367,7081071}.

Reconfigurable intelligent surfaces (RIS) have recently been proposed as a tool for enhancing wireless security at the physical layer. An RIS is an array of passive elements whose reflection coefficients can be adjusted by means of control signals (e.g., a biasing voltage on a varactor diode) fed to each element. With sufficient channel state information (CSI), the signals reflected by the RIS can provide ``virtual'' line-of-sight paths around blockages, they can be constructively superimposed with direct path signals to enhance the desired signal power, or they can be destructively combined with the direct paths to mitigate the impact of multiuser interference~\cite{9475160,9086766,9424177}. Using the same principles, recent research has also demonstrated that RIS can be an effective auxiliary means to resist physical layer threats in wireless systems, including jamming attacks~\cite{9198898,10143420,9779086,9439833,9385957,10175566}. 

However, RIS technology can also be used for malicious purposes when deployed and controlled by an adversary. This can include an attack on the controller of a legitimate RIS, causing it to behave in an unwanted way \cite{9941040}, or placing an RIS to either increase leakage of a desired signal towards an eavesdropper, or increase the jamming power at certain receivers \cite{wei2023metasurface,9789438,9605003,10073942}. In another type of attack, an adversary can control an RIS in such a way that the direct path signal is destructively combined with the RIS reflection to reduce the signal-to-noise ratio (SNR) or completely cancel the signal at a receiver \cite{9112252}. The above mentioned methods require that the adversary possess CSI for the links it intends to attack, and thus are difficult to implement in practice. Other methods that do not require CSI have been proposed for disrupting time-division duplex (TDD) multiuser multi-input single-output (MU-MISO) systems, in which uplink pilot data is used to estimate the CSI for the design of downlink precoders. For example, in~\cite{9186205}, a pilot contamination attack is implemented on a two-cell TDD MISO system, where an adversarial RIS employs one set of (random) phase shifts during the uplink transmission when the channel is estimated, and a different set of phase shifts during the downlink. This degrades the downlink performance since the precoder is no longer optimal for the modified downlink channel.
In \cite{10081025}, the same attack strategy is applied to a MU-MISO system. The downlink performance is severely degraded due to the increased multiuser interference caused by the inaccurate precoder design for the modified downlink channel. A similar approach is used in \cite{10149173}, except that the attacking RIS is silent during the uplink, and then employs time-varying phase shifts during the downlink. The authors refer to these techniques as ``passive jamming'' attacks, since they disrupt communications without expending any power and thus are difficult to detect.

A significant drawback of \cite{9186205,10081025,10149173} is that the malicious RIS must be synchronized with the uplink training and downlink data transmission phases of the MU-MISO network under attack. In addition, rapid changes in the CSI at a rate much faster than the normal channel coherence time can be easily detected. In this paper, we present a novel method that eliminates these drawbacks, while still requiring no CSI and nor active transmissions. The proposed approach relies on the use of a so-called ``non-diagonal'' RIS \cite{9737373,10302331} (ND-RIS), in which the signal received by one RIS element can be reflected through another RIS element, after the application of a controllable phase shift. The term ``non-diagonal'' refers to the fact that, unlike a conventional RIS whose reflection coefficients appear as a diagonal matrix in the cascaded channel model, the architecture of \cite{9737373,10302331} leads to a non-diagonal and non-symmetric phase shift matrix. From an attacker's perspective, the key benefit of this type of RIS is that it leads to non-reciprocal propagation when the phase-shift matrix is non-symmetric. A notional circuit architecture for designing an ND-RIS was presented in \cite{9737373}. The ND-RIS design is related to a recent more broad class of metasurfaces referred to as ``beyond diagonal'' RIS (BD-RIS), which also provides more general structures for the phase shift matrix \cite{9913356,10159457,10187688,10197228,10237233,10316535,10302331} with additional degrees of freedom. While the BD-RIS architectures suggested in this recent work are constrained to be reciprocal (i.e., for use in traditional networks where reciprocity is desired), non-reciprocal versions can be constructed with the addition of isolators or circulators in the BD-RIS impedance network.

The non-reciprocity of an ND-RIS can be exploited by simply placing it in an environment where it will capture sufficient energy in the uplink of a TDD-based MU-MISO system. The ND-RIS will inherently and passively make the uplink channel different from the downlink channel without requiring any CSI nor synchronization with the MU-MISO system, and without continually changing its state. A precoder designed based on the estimated uplink channel will thus be suboptimal for the non-reciprocal downlink channel, and the performance of the MU-MISO system will be degraded due to the resulting increase in multiuser interference. While CSI is not required to launch a non-reciprocity attack such as the one proposed, CSI can be exploited when available to maximize the effectiveness of the attack.
The contributions of the paper are summarized below.
\begin{itemize}
\item We propose a novel Channel Reciprocity AttaCK (CRACK) on TDD-based MU-MISO systems using an ND-RIS. When the ND-RIS is located such that sufficient energy from the users is reflected by the ND-RIS to the base station (BS), the ND-RIS destroys the reciprocity of the channel and leads to a significant increase in multiuser interference. The method is passive (no transmit power required), it can be implemented without any CSI for the legitimate system, it does not require synchronization with the network under attack, and it can remain static without continually changing its state. CRACK is also suitable for communication systems under any type of channel model, and we demonstrate the effectiveness of the CRACK approach using an extensive set of simulation studies. %to launch severe communication attacks on legitimate users (LUs) without relying on both jamming power and LUs' detailed knowledge. Concretely, an adversarial non-diagonal Reconfigurable Intelligent Surface (ND-RIS) with adjustable phase shifts is deployed to actively deteriorate the downlink transmission in multi-user multiple-input single-output (MU-MISO) systems, where the ND-RIS acts like a malicious passive jammer to interfere with the entire communication process, including channel estimation and data transmission, by breaking the channel reciprocity between the uplink and downlink.
\item We conduct a theoretical analysis of the downlink ergodic sum rate achieved by the MU-MISO system under CRACK, for the case where the BS adopts maximum ratio transmit (MRT) beamforming for Rician RIS-user channel models, Rayleigh BS-user channel models, and a line-of-sight (LoS) BS-RIS channel model. Our simulation results demonstrate that the derived expression accurately predicts the achieved ergodic sum rate. In addition, our simulations show that CRACK is even more effective against systems that employ precoding methods such as zero-forcing whose performance is very sensitive to assumptions about the CSI.
\item Using our derived expression for the ergodic sum rate, we design a genetic algorithm (GA)-based optimization method that determines the ND-RIS connections and phase shifts that minimize the sum rate assuming knowledge of the LoS channel components of the RIS-cascaded channels and global statistical CSI of the Rayleigh channel components. Our numerical results demonstrate that when such CSI is available, the algorithm applied together with CRACK can considerably reduce the performance of the MU-MISO network, not only for the MRT precoder but for the zero-forcing (ZF) case as well.
\end{itemize}

The remainder of the paper is organized as follows. In Section~\uppercase\expandafter{\romannumeral2}, the uplink and downlink channels of an MU-MISO system affected by the proposed ND-RIS CRACK are modeled, and the ergodic sum-rate performance metric used to quantify the impact of the attack is defined. Section~\uppercase\expandafter{\romannumeral3} provides a theoretical analysis of the proposed ND-RIS-based CRACK approach and notes some properties based on the derived analysis. Section~\uppercase\expandafter{\romannumeral4} introduces the GA-based optimization algorithm, and Section~\uppercase\expandafter{\romannumeral5} presents and discusses several representative simulation results. Finally, conclusions are drawn in Section~\uppercase\expandafter{\romannumeral6}.

\textit{Notation}: Vectors and matrices are denoted by boldface lower and upper case letters, respectively. The operators $(\cdot)^T$ and $(\cdot)^H$ represent the transpose and Hermitian transpose, respectively. The magnitude of the complex scalar $x$ and the $2$-norm of vector $\mathbf{x}$ are respectively denoted by $|x|$ and $\Vert \mathbf{x}\Vert$. The real part of scalar $x$ is given by ${\rm{Re}}\{x\}$, the trace of a matrix $\mathbf{X}$ is written as ${\rm{Tr}}\left(\mathbf{X}\right)$, and $diag(\mathbf{x})$ denotes a diagonal matrix whose diagonal elements are given by the elements in $\mathbf{x}$. The notation $\left[\mathbf{x}\right]_i$ represents the $i$-th element of vector $\mathbf{x}$ and $\left[\mathbf{X}\right]_{i,j}$ represents the element in row $i$ and column $j$ of matrix $\mathbf{X}$.

\section{System Model}\label{s2}

%\begin{table}[ht]
%       \vspace{-0.06in}
%		\centering
%        \caption{NOTATION AND PARAMETERS}
%        \label{tab:Notation_settings}
%		\begin{tabular}{|l|c|c|}\hline
%			Notation& Parameters \\\hline
%            $M$ & Antenna number of Base station\\\hline
%            $N$ & Element number of RIS \\\hline
%            $K$ & Number of LUs \\\hline
%            $G$ & Size of population \\\hline
%		$\mathbf{\Phi}$ & Phase shift matrix of RIS\\\hline
%			$\mathbf{J}_l$ and $\mathbf{J}_r$ & Permutation matrix of RIS \\\hline
%            $\kappa_{i,j}$ & Rician factor between $i$ and $j$ \\\hline
%            $\varphi_{i,j}^a$ and $\varphi_{i,j}^e$ & Azimuth and elevation AoA from $i$ to $j$\\\hline
%            $\phi_{i,j}^a$ and $\phi_{i,j}^e$ & Azimuth and elevation AoD from $i$ to $j$\\\hline
%			$\alpha_{i,j}$ & Path loss between $i$ and $j$\\\hline
%			$P_k$& Transmit power to $\rm{LU_k}$\\\hline
%            $\sigma^2$& Additional Gaussian noise power\\\hline
%		\end{tabular}
%		\vspace{-0.05in}
%\end{table}

\begin{figure}[!t]
\begin{center}
\includegraphics[width=2.8 in]{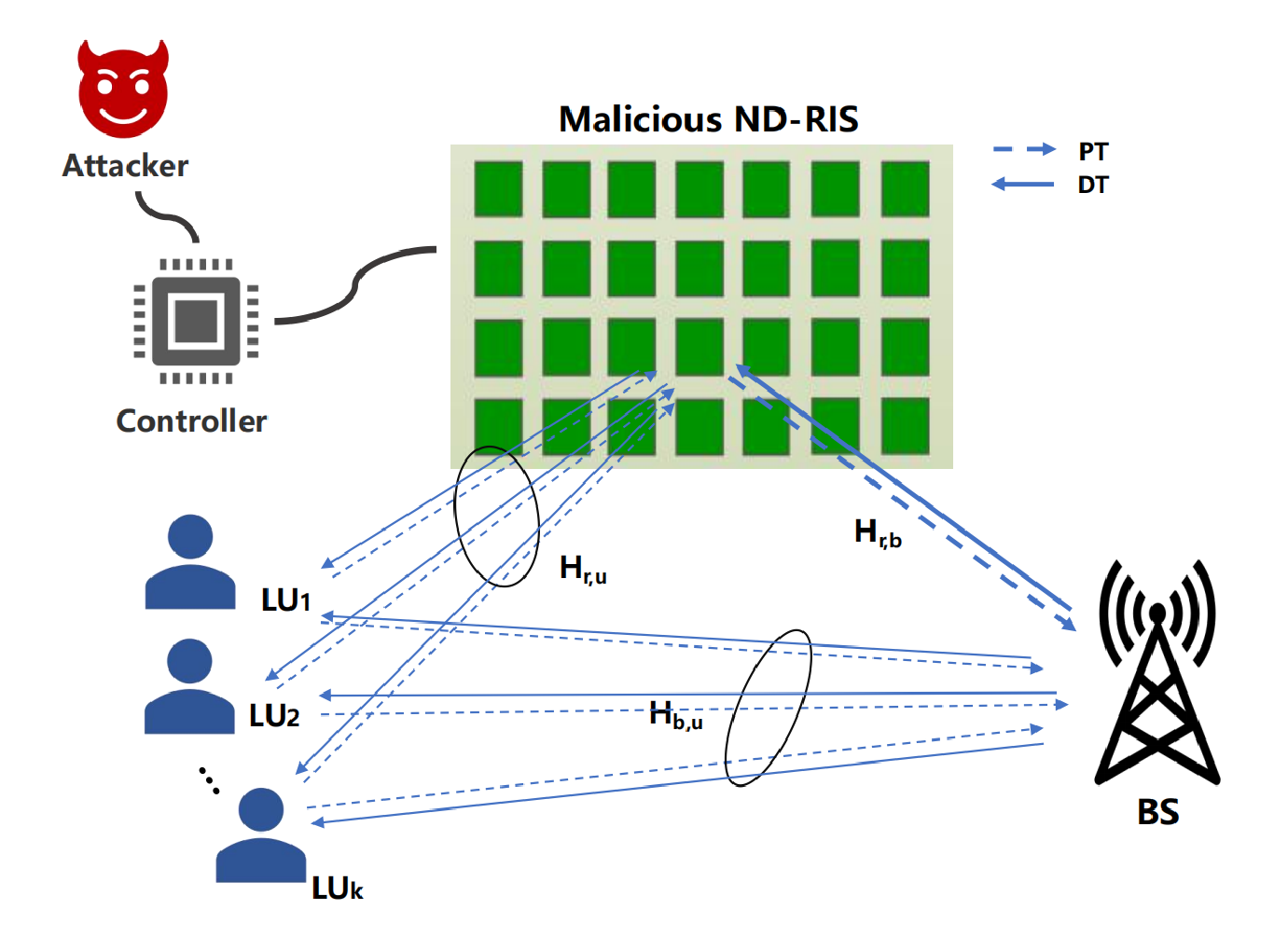}
\end{center}
\vspace{-0.08in}
\caption{Illustration of an MU-MISO system with malicious ND-RIS unknown to the BS that breaks the reciprocity of the uplink pilot transmission (PT) channel (dashed lines) and the downlink data transmission (DT) channel (solid lines).}\label{f1}
\vspace{-0.13in}
\end{figure}
Fig. \ref{f1} illustrates the considered MU-MISO system, where an $M$-antenna BS communicates with $K < M$ single-antenna legitimate users ($\rm{LUs}$) denoted as $\rm{LU_1}, \rm{LU_2},...,\rm{LU_K}$. In addition, an ND-RIS with $N$ elements is present and affects the propagation of signals to/from the BS, although the BS is unaware of its presence. The terms $\mathbf{h}_{k,r}, \mathbf{H}_{r,b}$ and $\mathbf{h}_{k,b}$, respectively denote the flat fading channels between $\rm{LU_k}$ and the ND-RIS, between the BS and the ND-RIS, and between $\rm{LU_k}$ and the BS. The channel $\mathbf{H}_{up}$ is the overall combined channel between the BS and users, including the ND-RIS, while $\mathbf{H}_{down}=\mathbf{H}_{up}^T$ is the channel that the BS {\em assumes} for the downlink channel under reciprocity. 

\subsection{MU-MISO Downlink Precoding}

In a TDD implementation, the BS first estimates the CSI using uplink pilot signals sent by the $\rm{LUs}$, and then, based on the assumption that the uplink and downlink channels are identical, the estimated CSI is used to design a precoder for the downlink. Assuming the pilot signals transmitted by $\rm{LU_k}$ at time $t$ are denoted by $s_{t,k}$, the received signal from all of the $\rm{LUs}$ at the BS is given by
\begin{equation}\label{e1}
\mathbf{y}_{t}=\sum_{k=1}^K \sqrt{p_k}\mathbf{h}_{k,b}^*s_{t,k}+\mathbf{n}_t,
\end{equation}
where $p_k$ is the power of the pilot signal assuming $E(|s_{t,k}|^2)=1$ and $\mathbf{n}_t$ denotes the receiver noise. Due to the existence of the ND-RIS, the uplink channel for $\rm{LU_k}$ is given by $\mathbf{h}^*_{k,b}=\mathbf{H}_{r,b}\mathbf{\Phi}^* \mathbf{h}_{k,r}+\mathbf{h}_{k,b}$, which includes the contribution of the direct path $\mathbf{h}_{k,b}$ and the cascaded ND-RIS path, where $\mathbf{\Phi}^*$ represents the impact of the ND-RIS. Based on the received pilot signals, the channels for all the users can be estimated by the BS, and we will assume that the overall uplink channel $\mathbf{H}_{up}=[\mathbf{h}^*_{1,b},\cdots,\mathbf{h}^*_{K,b}] = \mathbf{H}_{r,b}\mathbf{\Phi}^* \mathbf{H}_{r,u}+\mathbf{H}_{b,u}\in \mathbb{C}^{M \times K}$ is known perfectly at the BS, where $\mathbf{H}_{r,u}=[\mathbf{h}_{1,r},\cdots,\mathbf{h}_{K,r}]$ and $\mathbf{H}_{b,u}=[\mathbf{h}_{1,b},\cdots,\mathbf{h}_{K,b}]$ are the combined user-RIS channel and combined user-BS direct channel, respectively. {\em It is important to note that the BS can only estimate the overall channel $\mathbf{H}_{up}$ and not its constituent parts $\mathbf{H}_{b,u}$ and $\mathbf{H}_{r,b}\mathbf{\Phi}^* \mathbf{H}_{r,u}$, which are not identifiable in this application}. In a system with a cooperative RIS, these terms can be estimated separately, but only because the RIS changes its response in a known way during the training. In our application, the ND-RIS response remains constant and is unknown to the BS. This prevents the BS from somehow nulling out the signal from the ND-RIS.

Based on channel reciprocity, the BS will assume that the downlink channel is $\mathbf{H}_{down}=\mathbf{H}_{up}^T = \mathbf{H}_{r,u}^T\mathbf{\Phi}^{*T} \mathbf{H}_{r,b}^T+\mathbf{H}_{b,u}^T$. However, the {\em actual} downlink channel is given by $\mathbf{H}_{down}^*= \mathbf{H}_{r,u}^T\mathbf{\Phi}^* \mathbf{H}_{r,b}^T+\mathbf{H}_{b,u}^T$. For a conventional RIS, $\mathbf{H}_{down}$ and $\mathbf{H}^*_{down}$ are equal since the phase shift matrix of a conventional RIS is diagonal, and hence symmetric. However, this is not the case for an ND-RIS. In particular, the ND-RIS phase shift matrix can be denoted as $\mathbf{\Phi}^* =\mathbf{J} \mathbf{\Phi}$, where $\mathbf{\Phi}=\text{diag} \, [\theta_1,\cdots,\theta_N]\in \mathbb{C}^{N \times N}$ is a diagonal matrix with $\theta_i=\beta_i e^{j\phi_i}$ and $\mathbf{J}$ is a permutation matrix. In the application considered here, the goal is to maximize reflection from the malicious ND-RIS, and so we set $\beta_i=1$ in the sequel. Consequently $\theta_i$ has unit modulus with $\phi_i \in [0,2\pi)$. The permuation $\mathbf{J}$ scrambles the order of the elements of a given $1\times N$ vector $\mathbf{x}$ such that $\mathbf{x}\mathbf{J}=[x_{[1]},x_{[2]},\dots,x_{[N]}]$, and we can consider the permutation to be a mapping from one set of ordered indices into another: $\mathbf{J}: n\to [n]$. In general, $\mathbf{\Phi}^{*T} = \mathbf{\Phi} \mathbf{J} \ne  \mathbf{J \Phi} = \mathbf{\Phi}^{*}$, which destroys the channel reciprocity and enables CRACK.

The BS designs the downlink precoder $\mathbf{W}=[\mathbf{w}_1, \cdots, \mathbf{w}_K]$ based on the assumed downlink channel $\mathbf{H}_{down}$. For example, for the case of a maximum ratio transmit (MRT) beamformer, the precoder for $\rm{LU_k}$ is the conjugate transpose of the downlink channel for $\rm{LU_k}$, given by $\mathbf{w}_k=((\mathbf{h}^*_{k,b})^T)^H\in \mathbb{C}^{M \times 1}$. The overall precoder is therefore given by $\mathbf{W}=\mathbf{H}_{down}^H=\left(\mathbf{H}_{up}^T\right)^H$, and the signal vector received by the $K$ users is given by 
\begin{equation}\label{e2}
\mathbf{y}=\mathbf{H}^*_{down}\mathbf{W}\mathbf{D}\mathbf{s}+\mathbf{n}_d,
\end{equation}
where $\mathbf{D}=\text{diag}(\sqrt{P_1}, \sqrt{P_2},...,\sqrt{P_K})$ is a ${K \times K}$ diagonal power allocation matrix and $\mathbf{n}_d$ denotes the received noise vector, whose elements are modeled as independent and identically distributed (i.i.d.) Gaussian random variables with zero mean and variance $\sigma^2$. 

\subsection{Novelty of CRACK Approach}
As with the RIS attacks proposed in \cite{9186205,10081025,10149173}, the proposed CRACK approach can be implemented in an entirely passive way, without transmission of any signal, and in general it does not require any CSI for the legitimate system. However, CRACK has important advantages compared with these other methods. The approaches in \cite{9186205,10081025,10149173} rely on achieving $\mathbf{H}_{down} \ne \mathbf{H}_{up}^T$ by making the diagonal phase shift matrix of a conventional RIS change from one state during uplink training to another state during downlink precoding (or multiple times during each). This type of attack causes what is called ``active channel aging'' (ACA), which shortens the channel coherence time and creates a more rapidly varying channel than what the BS would normally assume. ACA attacks have two significant drawbacks compared with CRACK. First, an ACA attacker must be synchronized with the uplink and downlink transmission times of the legitimate TDD system, and second, such rapid variations in the channel are quite easy to detect. In the proposed CRACK approach, the channel's coherence time is not modified, and the presence of the ND-RIS is not revealed by any physically observable phenomenon, other than degraded performance at the users. It is also important to point out that, although we focus on the flat-fading case here for simplicity, CRACK is agnostic to the assumed channel model. It is equally effective for any type of flat- or frequency-selective fading model, and also for both near- and far-field scenarios

\subsection{Achievable Ergodic Rate for MRT Precoding}
As explained above, the actual channel from the BS to $\rm{LU_k}$, which we denote as the $1\times M$ vector $\mathbf{h}^*_{b,k}$, is given by $\mathbf{h}^*_{b,k}=\mathbf{h}^T_{k,r}\mathbf{\Phi}^* \mathbf{H}^T_{r,b}+\mathbf{h}^T_{k,b}$. Thus, the received downlink signal at $\rm{LU_k}$ is written as
\begin{equation}\label{e4}
y_k=\sqrt{P_k}\mathbf{h}^*_{b,k} \mathbf{w}_{k}s_k+\mathbf{h}^*_{b,k}\sum_{i=1, i\neq k}^{K}\sqrt{P_i}\mathbf{w}_{i}s_i + n_k.
\end{equation}
For MRT, the precoding vector for $\rm{LU_k}$ is 
\[
\mathbf{w}_{k}=((\mathbf{h}^*_{k,b})^T)^H =(\mathbf{H}^T_{r,b})^H\widetilde{\mathbf{\Phi}} (\mathbf{h}^T_{k,r})^H+(\mathbf{h}^T_{k,b})^H \; ,
\]
where $\widetilde{\mathbf{\Phi}}=((\mathbf{\Phi}^*)^T )^H=\mathbf{J} \mathbf{\Phi}^H $ is the conjugate of $\mathbf{\Phi}^*$. For notational simplicity, we write $\mathbf{h}_{r,k}=\mathbf{h}^T_{k,r}$, $\mathbf{H}_{b,r}=\mathbf{H}^T_{r,b}$, and $\mathbf{h}_{b,k}=\mathbf{h}^T_{k,b}$, so that the actual received signal at $\rm{LU_k}$ can be formulated as
\begin{equation}\label{e5}
\begin{aligned}
y_k&=\sqrt{P_k}\mathbf{h}^*_{b,k} (\mathbf{H}^H_{b,r}\widetilde{\mathbf{\Phi}} \mathbf{h}^H_{r,k}+\mathbf{h}^H_{b,k})s_k \\&+\mathbf{h}^*_{b,k}\sum_{i=1, i\neq k}^{K}\sqrt{P_i} (\mathbf{H}^H_{b,r}\widetilde{\mathbf{\Phi}} \mathbf{h}^H_{r,i}+\mathbf{h}^H_{b,i})s_i  + n_k.
\end{aligned}
\end{equation}

The downlink signal-to interference-plus-noise ratio (SINR) for $\rm{LU_k}$ is then given by (\ref{e6}).
\begin{comment}
\begin{equation}\label{e6}
\eta_k=\frac{P_k|\mathbf{h}_{b,k}\mathbf{w}_{k}|^2}{\sum_{i=1, i\neq k}^{K}P_i|\mathbf{h}_{b,k}\mathbf{w}_{i}|^2+\sigma^2}
\end{equation}
\end{comment}
\begin{figure*}
\begin{equation}\label{e6}
\eta_k=\frac{P_k|(\mathbf{h}_{r,k}\mathbf{\Phi}^* \mathbf{H}_{b,r}+\mathbf{h}_{b,k})(\mathbf{H}^H_{b,r}\widetilde{\mathbf{\Phi}} \mathbf{h}^H_{r,k}+\mathbf{h}^H_{b,k})|^2}{\sum_{i=1, i\neq k}^{K}P_i|(\mathbf{h}_{r,k}\mathbf{\Phi}^* \mathbf{H}_{b,r}+\mathbf{h}_{b,k})(\mathbf{H}^H_{b,r}\widetilde{\mathbf{\Phi}} \mathbf{h}^H_{r,i}+\mathbf{h}^H_{b,i})|^2+\sigma^2}.
\end{equation}
\hrulefill
\end{figure*}
Assuming that all channels are ergodic, the achievable ergodic rate for $\rm{LU_k}$ can be written as (\ref{e7}),
\begin{figure*}
\begin{equation}\label{e7}
\begin{aligned}
r_k^* &= \mathbb{E}[\log(1+\eta_k)]= \mathbb{E}\left[\log\left(1+\frac{P_k|(\mathbf{h}_{r,k}\mathbf{\Phi}^* \mathbf{H}_{b,r}+\mathbf{h}_{b,k})(\mathbf{H}^H_{b,r}\widetilde{\mathbf{\Phi}} \mathbf{h}^H_{r,k}+\mathbf{h}^H_{b,k})|^2}{\sum_{i=1, i\neq k}^{K}P_i|(\mathbf{h}_{r,k}\mathbf{\Phi}^* \mathbf{H}_{b,r}+\mathbf{h}_{b,k})(\mathbf{H}^H_{b,r}\widetilde{\mathbf{\Phi}} \mathbf{h}^H_{r,i}+\mathbf{h}^H_{b,i})|^2+\sigma^2}\right)\right]\\
&\overset{(a)}{\approx}\log\left(1+\frac{\mathbb{E}\left[P_k|(\mathbf{h}_{r,k}\mathbf{\Phi}^* \mathbf{H}_{b,r}+\mathbf{h}_{b,k})(\mathbf{H}^H_{b,r}\widetilde{\mathbf{\Phi}} \mathbf{h}^H_{r,k}+\mathbf{h}^H_{b,k})|^2\right]}{\mathbb{E}\left[\sum_{i=1, i\neq k}^{K}P_i|(\mathbf{h}_{r,k}\mathbf{\Phi}^* \mathbf{H}_{b,r}+\mathbf{h}_{b,k})(\mathbf{H}^H_{b,r}\widetilde{\mathbf{\Phi}} \mathbf{h}^H_{r,i}+\mathbf{h}^H_{b,i})|^2\right]+\sigma^2}\right),
\end{aligned}
\end{equation}
\hrulefill
\end{figure*}
where the approximation in step $(a)$ has been shown to be accurate for massive MIMO systems \cite{6816003,9355404,9935294,10034679,9973364,9973349}.
%can be proved by \cite[Lemma~1]{6816003}. Specifically, in massive MIMO systems, due to a large number of BS antennas and RIS elements, this approximation results will be particularly accurate \footnote{This Lemma has been widely adopted by the ergodic rate analysis of RIS-aided or MIMO systems~\cite{6816003,9355404,9935294,10034679,9973364,9973349}. The detailed proof process of it can be found in \cite{6816003}.}.

\section{Achievable Ergodic Rate Analysis Under ND-RIS CRACK}\label{s3}
In this section, we derive an approximate closed-form
expression for the achievable ergodic rate of the MU-MISO system under the ND-RIS-based channel reciprocity attack for a fixed choice of $\mathbf{\Phi}^*$. The theoretical results reveal the impact of various system variables, including the number of antennas at the BS, the number of reflecting elements at the RIS, the transmit power, and the path loss exponent. We will also present asymptotic expressions for some special cases. For more compact notation, (\ref{e7}) can be converted into (\ref{e13}), 
\begin{figure*}
\begin{equation}\label{e13}
\begin{aligned}
r_k^* &\approx\log\left(1+\frac{\mathbb{E}\left[P_k|(\mathbf{g}_{k}+\mathbf{h}_{b,k})(\mathbf{q}_{k}+\mathbf{h}^H_{b,k})|^2\right]}{\mathbb{E}\left[\sum_{i=1, i\neq k}^{K}P_i|(\mathbf{g}_{k}+\mathbf{h}_{b,k})(\mathbf{q}_{i}+\mathbf{h}^H_{b,i})|^2\right]+\sigma^2}\right)\\
&\approx\log\left(1+\frac{\mathbb{E}\left[P_k|(\mathbf{g}_{k}+\mathbf{h}_{b,k})(\mathbf{q}_{k}+\mathbf{h}^H_{b,k})|^2\right]}{\sum_{i=1, i\neq k}^{K}P_i \left(\mathbb{E}\left[|\mathbf{g}_{k}\mathbf{q}_{i}|^2\right]+\mathbb{E}\left[|\mathbf{g}_{k}\mathbf{h}^H_{b,i}|^2\right] +\mathbb{E}\left[|\mathbf{h}_{b,k}\mathbf{q}_{i}|^2\right] + \mathbb{E}\left[|\mathbf{h}_{b,k}\mathbf{h}^H_{b,i}|^2\right] \right)+\sigma^2}\right),
\end{aligned}
\end{equation}
\hrulefill
\end{figure*}
where $\mathbf{g}_{k}=(\mathbf{h}_{r,k}\mathbf{\Phi}^* \mathbf{H}_{b,r})$, $\mathbf{q}_{k}=(\mathbf{H}^H_{b,r}\widetilde{\mathbf{\Phi}} \mathbf{h}^H_{r,k})$, $\mathbf{q}_{i}=(\mathbf{H}^H_{b,r}\widetilde{\mathbf{\Phi}} \mathbf{h}^H_{r,i})$. After describing the assumed channel models in the next section, we evaluate the expectations in~\eqref{e13} and derive a closed-form expression for the ergodic rate in Section~\ref{sec:aer}.

\subsection{Assumed Channel Models}
%Therefore, it is highly possible that strong LoS components exist in RIS-BS channels once it is installed on the facades of a tall building or deployed in the near-field of the BS~\cite{9355404,9754563,8964330,10012964,9066923}
For the analysis that follows, we will assume the following common channel models: Rayleigh fading model for all BS-User direct links, Rician model for User-RIS links, Line-of-sight (LoS) model for the BS-RIS link. Specifically, the channel between $\rm{LU_k}$ and BS is expressed as
\begin{equation}\label{e001}
\mathbf{h}_{k,b}=\sqrt{\alpha_{k,b}}\widetilde{\mathbf{h}}_{k,b},
\end{equation}
where  $\alpha_{k,b}=\rho d_{k,b}^{-\iota_{k,b}}$ is the distance-dependent large-scale path-loss factor, $\rho$ is the path loss at the reference distance $d_0=1$m, $d_{k,b}$ is the distance between the  $\rm{LU_k}$ and BS, and $\iota_{k,b}$ is the path loss exponent of the  $\rm{LU_k}$-to-BS link. The elements of $\widetilde{\mathbf{h}}_{k,b}  \in \mathbb{C}^{M \times 1}$ follow an i.i.d. complex Gaussian distribution with zero mean and unit variance. The channel $\mathbf{h}_{k,r}$ between user $k$ and the ND-RIS is described as
\begin{equation}\label{e9}
\mathbf{h}_{k,r}=\sqrt{\alpha_{k,r}}\left(\sqrt{\frac{\kappa_{k,r}}{1+\kappa_{k,r}}} \overline{\mathbf{h}}_{k,r}+\sqrt{\frac{1}{1+\kappa_{k,r}}} \widetilde{\mathbf{h}}_{k,r}\right) , 
\end{equation}
where $\kappa_{k,r}$ is the Rician factor, $\overline{\mathbf{h}}_{k,r}$ is the line-of-sight (LoS) channel component which is assumed to satisfy $\overline{\mathbf{h}}_{k,r}^H\overline{\mathbf{h}}_{k,r}=N$, and the elements of the non-LoS (NLoS) component $\widetilde{\mathbf{h}}_{k,r}$ follow an i.i.d. complex Gaussian distribution with zero mean and unit variance. Finally, the LoS BS-RIS channel $\mathbf{H}_{r,b}$ is described by the deterministic rank-one matrix
\begin{equation}\label{e10}
\mathbf{H}_{r,b}=\sqrt{\alpha_{r,b}}\overline{\mathbf{H}}_{r,b}=\sqrt{\alpha_{r,b}}\mathbf{a}_N \mathbf{a}_M^H,
\end{equation}
where $\mathbf{a}_N^H\mathbf{a}_N=N$ and $\mathbf{a}_M^H\mathbf{a}_M=M$ represent the corresponding array response vectors.

In the simulations, we will also include results for the more general case where the BS-RIS link is Rician, in which case the channel model becomes
%\footnote{The RIS should be deployed near the BS to reduce the multiplicative fading effect~\cite{9427474}. 
\begin{equation}\label{e11}
\mathbf{H}_{r,b}=\sqrt{\alpha_{r,b}}\left(\sqrt{\frac{\kappa_{r,b}}{1+\kappa_{r,b}}} \overline{\mathbf{H}}_{r,b}+\sqrt{\frac{1}
{1+\kappa_{r,b}}} \widetilde{\mathbf{H}}_{r,b}\right),
\end{equation}
where $\kappa_{r,b}$ is the Rician factor, and the elements of the non-LoS (NLoS) component $\widetilde{\mathbf{H}}_{r,b} \in \mathbb{C}^{M \times N}$ follow an i.i.d. complex Gaussian distribution with zero mean and unit variance. 

%%\begin{equation}\label{e12}
%\mathbf{a}^H_i \mathbf{a}_i=i,
%\end{equation}
%where $\mathbf{a} \in \mathbb{C}^{i \times 1}$ is a general expression of the antenna array response and $i$ is the size of antenna array. In the following analysis, the LoS components $\overline{\mathbf{H}}_{b,r}$ is denoted as\footnote{In the simulation, we adopt a uniform liner array (ULA) model for BS and ND-RIS. Nevertheless, the analysis is applicable for arbitrary antenna models satisfying (\ref{e12}), e.g.,  uniform rectangular array (URA).}
%To avoid confusion, we still use $\overline{\mathbf{h}}_{k,r}$ to represent the antenna array response between RIS and $\rm{LU_k}$, which also satisfies the property of (\ref{e12}).

\subsection{Achievable Ergodic Rate}\label{sec:aer}
%Simulation results indicate that ND-RIS design may bring a more serious attack threat when the LoS path between ND-RIS and users exists. Based on the previous communication scenario, we assume there exist LoS links between Users and RIS after the appropriate deployment of ND-RIS. Moreover, the weak influence of NLoS links on the RIS-BS channel is ignored in this model for ease of analysis\footnote{Considering the limited reflection of mmWave and strong LoS component existing in RIS-BS channel after appropriate deployment, we can ignore the weak influence of NLoS components~\cite{9935294,8964330}.}, where $\mathbf{H}_{r,b}=\sqrt{\alpha_{r,b}} \overline{\mathbf{H}}_{r,b}$ ($\kappa_{r,b}\rightarrow \infty$). 

The following theorem provides an approximate closed-form expression for the downlink ergodic rate of the system under CRACK for the case of an ND-RIS configured with phase response $\mathbf{\Phi}^*$.
\begin{theorem}\label{t1}
Under ND-RIS CRACK, assuming a LoS BS-RIS channel, Rayleigh BS-user channels and Rician RIS-user channels, the downlink ergodic rate for $\rm{LU_k}$ can be approximated as
\begin{equation}\label{e14}
\begin{aligned}
r_k^* &\approx\log\left(1+\frac{P_k\mathscr{E}_k^{signal}(\mathbf{\Phi}^*)}{\sum_{i=1, i\neq k}^{K}P_i\mathscr{I}_{k,i}(\mathbf{\Phi}^*)+\sigma^2}\right), 
\end{aligned}
\end{equation}
where $\mathscr{E}_k^{signal}(\mathbf{\Phi}^*)$ and $\mathscr{I}_{k,i}(\mathbf{\Phi}^*)$ are respectively given as (\ref{e17}) and (\ref{e18})
\begin{figure*}
\begin{equation}\label{e17}
\begin{aligned}
\mathscr{E}_k^{signal}(\mathbf{\Phi}^*) &= \left(\mathbb{E}\left[\left|\mathbf{g}_{k}\mathbf{q}_{k}\right|^2\right]+ 2\mathbb{E}\left[
 {\rm{Re}}\left\{\mathbf{g}_{k}\mathbf{q}_{k}\mathbf{h}_{b,k}\mathbf{h}^H_{b,k}\right\}\right] +\mathbb{E}\left[\left|\mathbf{g}_{k}\mathbf{h}^H_{b,k}\right|^2\right] \right. \\ 
 & \qquad\left. +2\mathbb{E}\left[{\rm{Re}}\left\{\mathbf{g}_{k}\mathbf{h}^H_{b,k}\mathbf{q}^H_{k}\mathbf{h}^H_{b,k}\right\}\right]+ \mathbb{E}\left[\left|\mathbf{h}_{b,k}\mathbf{q}_{k}\right|^2\right] +\mathbb{E}\left[\left|\mathbf{h}_{b,k}\mathbf{h}^H_{b,k}\right|^2\right] \right)\\
 &= c^2_k \left(\kappa^2_{k,r}M^2\left|f_k(\mathbf{\Phi}^*)\right|^2 \left|f^*_k(\mathbf{\Phi}^*)\right|^2+\kappa_{k,r}M^2 N\left|f_k(\mathbf{\Phi}^*)\right|^2 +\kappa_{k,r}M^2N \left|f^*_k(\mathbf{\Phi}^*)\right|^2 \right. \\
& \qquad \left.+M^2\Psi_1(\mathbf{\Phi}^*)\right) 
+2c_k \alpha_{k,b}M^2\Bigl(\kappa_{k,r} {\rm{Re}}\left\{f_k(\mathbf{\Phi}^*)f^*_k(\mathbf{\Phi}^*) \right\}+{\rm{Re}}\left\{{\rm{Tr}}\left(\mathbf{t}_1(\mathbf{\Phi}^*)\right)\right\}\Bigr)  \\ & \qquad+\alpha^2_{k,b}(M^2+M) + c_k\alpha_{k,b}\left(\kappa_{k,r} M\left|f_k(\mathbf{\Phi}^*)\right|^2+MN\right) +c_k\alpha_{k,b}\left(\kappa_{k,r} M\left|f^*_k(\mathbf{\Phi}^*)\right|^2+MN\right),
\end{aligned}
\end{equation}
\hrulefill
\end{figure*}
\begin{figure*}
\begin{equation}\label{e18}
\begin{aligned}
\mathscr{I}_{k,i}(\mathbf{\Phi}^*) &=\mathbb{E}\left[|\mathbf{g}_{k}\mathbf{q}_{i}|^2\right]+\mathbb{E}\left[|\mathbf{g}_{k}\mathbf{h}^H_{b,i}|^2\right] +\mathbb{E}\left[|\mathbf{h}_{b,k}\mathbf{q}_{i}|^2\right] + \mathbb{E}\left[|\mathbf{h}_{b,k}\mathbf{h}^H_{b,i}|^2\right] \\
&=c_k c_i \left(\kappa_{k,r}\kappa_{i,r}M^2\left|f_k(\mathbf{\Phi}^*)\right|^2 \left|f_i(\mathbf{\Phi}^*)\right|^2+\kappa_{k,r}M^2 N\left|f_k(\mathbf{\Phi}^*)\right|^2+\kappa_{i,r}M^2N \left|f_i(\mathbf{\Phi}^*)\right|^2 \right. \\
& \qquad \left. +M^2 {\rm{Tr}}\left\{\Psi_2\Psi^H_2\right\}\Bigr)+c_k\alpha_{i,b}\left(\kappa_{k,r} M\left|f_k(\mathbf{\Phi}^*)\right|^2+MN\right)  \right.  +c_i\alpha_{k,b}\left(\kappa_{i,r} M\left|f_i(\mathbf{\Phi}^*)\right|^2+MN\right)+\alpha_{k,b}\alpha_{i,b}M,
\end{aligned}
\end{equation}
\hrulefill
\end{figure*}
with various constant terms defined as follows:
\begin{equation}
    c_k\triangleq\frac{\alpha_{r,b}\alpha_{k,r}}{(1+\kappa_{k,r})} \quad c_i\triangleq\frac{\alpha_{r,b}\alpha_{i,r}}{(1+\kappa_{i,r})} \quad f_k(\mathbf{\Phi}^*)\triangleq \overline{\mathbf{h}}_{r,k}\mathbf{\Phi}^*\mathbf{a}_N
\end{equation}
\begin{equation}
\begin{aligned}
    f^*_k(\mathbf{\Phi}^*)\triangleq \mathbf{a}^H_N \widetilde{\mathbf{\Phi}}\overline{\mathbf{h}}^H_{r,k} \qquad f_i(\mathbf{\Phi}^*)\triangleq\mathbf{a}^H_N \widetilde{\mathbf{\Phi}}\overline{\mathbf{h}}^H_{r,i},
    \end{aligned}
\end{equation}
\begin{equation}
    \mathbf{t}_1(\mathbf{\Phi}^*)\triangleq\mathbf{\Phi}^*\mathbf{a}_N \mathbf{a}_N^H \widetilde{\mathbf{\Phi}} \qquad \quad \Psi_2\triangleq \mathbf{a}_N \mathbf{a}_N^H \mathbf{J}.
\end{equation}
\begin{equation}\label{e005}
\begin{aligned}
\Psi_1(\mathbf{\Phi}^*)&\triangleq \sum_{i=1}^{N}2|\left[\mathbf{t}_1(\mathbf{\Phi}^*)\right]_{i,i}|^2+\sum_{i=1}^{N}\sum_{j=1, j\neq i}^{N}|\left[\mathbf{t}_1(\mathbf{\Phi}^*)\right]_{i,j}|^2\\&+\sum_{i=1}^{N-1}\sum_{j=i+1}^{N}2{\rm{Re}}\left\{\left[\mathbf{t}_1(\mathbf{\Phi}^*)\right]_{i,i}\left[\mathbf{t}_1(\mathbf{\Phi}^*)\right]^H_{j,j}\right\},
\end{aligned}
\end{equation} 
\end{theorem}
\begin{proof}
Please refer to the Appendix. $\hfill\blacksquare$
\end{proof}

Given the above result, we can obtain the sum ergodic rate of the multi-user system as $R=\sum_{k=1}^{K}r_k^*$ under ND-RIS-based CRACK, which is a function of $\mathbf{\Phi}$ and $\mathbf{J}$. In the following corollaries, we examine two limiting special cases that illustrate the effectiveness of the proposed attack.

\begin{corollary}
When the number of BS antennas is very large ($M \rightarrow \infty$), the sum ergodic rate will converge to a constant value that is independent of $M$ and determined by the design of the ND-RIS. In particular, the asymptotic ergodic rate of $\rm{LU_k}$ for this case can be approximated as
\begin{equation}\label{e19}
\begin{aligned}
r_k^* &\approx\log\left(1+\frac{P_k\mathscr{E}_k^*(\mathbf{\Phi}^*)}{\sum_{i=1, i\neq k}^{K}P_i\mathscr{I}^*_{k,i}(\mathbf{\Phi}^*)}\right),
\end{aligned}
\end{equation} 
where $\mathscr{E}_k^*(\mathbf{\Phi}^*)$ and $\mathscr{I}^*_{k,i}(\mathbf{\Phi}^*)$ are denote as (\ref{e003}) and (\ref{e004}), respectively.
\begin{figure*}
\begin{equation}\label{e003}
\begin{aligned}
\mathscr{E}_k^*(\mathbf{\Phi}^*)=c^2_k M^2\left(\kappa^2_{k,r}\left|f_k(\mathbf{\Phi}^*)\right|^2 \left|f^*_k(\mathbf{\Phi}^*)\right|^2+\kappa_{k,r} N\left|f_k(\mathbf{\Phi}^*)\right|^2 +\kappa_{k,r}N \left|f^*_k(\mathbf{\Phi}^*)\right|^2+\Psi_1(\mathbf{\Phi}^*)\right)\\
+2c_kM^2 \alpha_{k,b}\left(\kappa_{k,r} {\rm{Re}}\left\{f_k(\mathbf{\Phi}^*)f^*_k(\mathbf{\Phi}^*) \right\}+{\rm{Re}}\left\{{\rm{Tr}}\left(\mathbf{t}_1(\mathbf{\Phi}^*)\right)\right\}\right) +M^2\alpha^2_{k,b},
\end{aligned}
\end{equation} 
\hrulefill
\end{figure*}
\begin{figure*}
\begin{equation}\label{e004}
\begin{aligned}
\mathscr{I}^*_{k,i}(\mathbf{\Phi}^*)=c_k c_i M^2\left(\kappa_{k,r}\kappa_{i,r}\left|f_k(\mathbf{\Phi}^*)\right|^2 \left|f_i(\mathbf{\Phi}^*)\right|^2
+\kappa_{k,r} N\left|f_k(\mathbf{\Phi}^*)\right|^2+\kappa_{i,r}N \left|f_i(\mathbf{\Phi}^*)\right|^2+ {\rm{Tr}}\left\{\Psi_2\Psi^H_2\right\}\right),
\end{aligned}
\end{equation} 
\hrulefill
\end{figure*}
This illustrates that the BS cannot overcome CRACK by increasing the number of its antennas.
\end{corollary}
\begin{proof}
    When $M \rightarrow \infty$, the terms that contain $M^2$ in (\ref{e17}) and (\ref{e18}) will dominate, and thus the remaining terms can be neglected. Retaining the terms with $M^2$, we can derive~(\ref{e003}) and (\ref{e004}). The scaling by $M^2$ cancels in the ratio within the logarithm of~(\ref{e19}), and thus for large $M$ the sum ergodic rate will converge to a constant value that depends on the design of the ND-RIS but not the number of BS antennas $M$. $\hfill\blacksquare$
\end{proof}

\begin{corollary}
Assume equal power transmission to all users, i.e., $P_k=P$ for all $k$. As $P\rightarrow \infty$, the ergodic rate converges to a constant that is independent of $P$. Thus, the BS cannot overcome CRACK by increasing its transmit power.
\end{corollary}
\begin{proof}
When $P_k = P \rightarrow \infty$, the ergodic rate of $\rm{LU_k}$ in (\ref{e14}) becomes
\begin{equation}\label{e002}
\begin{aligned}
r_k^* &\approx\log\left(1+\frac{\mathscr{E}_k^{signal}(\mathbf{\Phi}^*)}{\sum_{i=1, i\neq k}^{K}\mathscr{I}_{k,i}(\mathbf{\Phi}^*)}\right), 
\end{aligned}
\end{equation}
which is a fixed value determined by the design of ND-RIS and is independent of $P$. $\hfill\blacksquare$
\end{proof}

\section{Genetic Algorithm Optimization for CRACK}

The expression derived in the previous section provides the ergodic sum rate assuming a fixed ND-RIS configuration $\mathbf{\Phi}^*$. Later we will see in the simulation section that, even without any CSI to design $\mathbf{\Phi}^*$, CRACK can significantly degrade the achievable sum rate of the MU-MISO system. On the other hand, if the ND-RIS does possess statistical CSI and information about the LoS components of the MU-MISO links, the choice of $\mathbf{\Phi}^*$ can be optimized in order to inflict greater performance loss. Here we address the problem of optimizing CRACK to solve the following sum ergodic rate minimization problem:
\begin{align}
&\min_{\mathbf{\Phi}, \mathbf{J}} \left\{\sum_{k=1}^{K}r^*_k\right\}  \quad \text{such that} \label{e20}\\
&C_1:\sum_{i=1}^{N}\left[\mathbf{J}\right]_{i,j} = 1,\sum_{j=1}^{N}\left[\mathbf{J}\right]_{i,j} = 1, \left[\mathbf{J}\right]_{i,j}\in \left\{0,1\right\}\tag{\ref{e20}{a}}\label{e20a},\\
%&C_2:\sum_{i=1}^{N}\left[\mathbf{J}_r\right]_{i,j} = 1,\sum_{j=1}^{N}\left[\mathbf{J}_r\right]_{i,j} = 1, \left[\mathbf{J}_r\right]_{i,j}\in \left\{0,1\right\}\tag{\ref{e20}{b}}\label{e20b},\\
%&C_2:R^m_S\ge R_{min}^m, \forall m\in\mathcal{M}, \tag{\ref{YY}{b}}\label{YYb}\\
&C_2:|\theta_n|=1, \forall n=1,2,...,N \tag{\ref{e20}{b}}\label{e20c},
\end{align}
where constraint \eqref{e20a} forces $\mathbf{J}$ to be a permutation matrix and \eqref{e20c} constrains the modulus of the RIS reflection coefficients to be unity. This is a non-convex mixed-integer programming problem that is very difficult to solve in a computationally efficient manner. Fortunately, since the problem depends on quantities that either do not change or change only very slowly, e.g., LoS angles and statistical CSI, the optimization only needs to be performed infrequently. 

The Genetic Algorithm (GA)~\cite{mirjalili2019evolutionary} is a good candidate for solving difficult non-convex problems with many local minima, especially those that involve integer-valued variables. For this reason, GA has been applied to RIS-related problems before~\cite{9743440,9366346,10355858,10326460}, and it is particularly well suited for cases where the RIS phases are quantized and must be chosen from a discrete set. However in our problem~\eqref{e20}, the integer variables apply to the construction of the optimal permutation matrix, whose elements are constrained to one per row and column. This is quite unlike the case of quantized phases, where the choice of phase at one RIS element does not constrain the choice at another, and requires the development of a specialized version of GA as described below.

GA is a type of evolutionary algorithm inspired by the process of natural selection. It works by iteratively evolving a population of candidate solutions over multiple generations. The process involves selecting individuals with higher ``fitness'' (to be defined), combining their genetic information through crossover, introducing random changes through mutation, and repeating the process until satisfactory solutions are obtained. The benefits of 
genetic-type algorithms include:
\begin{itemize}
\item {\bf Global Optimization:}  GA is capable of finding solutions in large search spaces by encouraging diversity in the search space, which helps prevent premature convergence to sub-optimal solutions, and hence increase the chance of finding the global optimum or near-optimal solutions.
\item {\bf Flexibility:} GA can be applied to a wide range of problems across different domains. It is not limited to specific mathematical functions or problem types, making it particularly applicable in complicated scenarios.
\item {\bf Mathematical Convenience:} Unlike traditional optimization techniques that require derivatives of the objective function, GA does not depend on gradient information. It is effective in problems where derivatives are difficult or impossible to obtain. Moreover, it is well-suited for solving combinatorial problems involving discrete elements and constraints.
\item {\bf Adaptability to Changing Environments:} GA can adapt to changes in the problem environment over time. As the fitness landscape changes, GA can adjust and search for new solutions.
\end{itemize}

The optimization problem in (\ref{e20}) is well suited for GA due to the integer constraints that arise from the permutation matrix, which make the problem difficult to solve via conventional optimization methods. GA can effectively handle combinatorial and non-convex optimization problems such as~\eqref{e20} since it does not rely on gradient information. The main steps for implementing GA to solve~\eqref{e20} are summarized in Fig.~\ref{f2}, and details are provided below and in Algorithm~1.

%\begin{comment}
\begin{figure}[!t]
\begin{center}
\includegraphics[width=2.3 in]{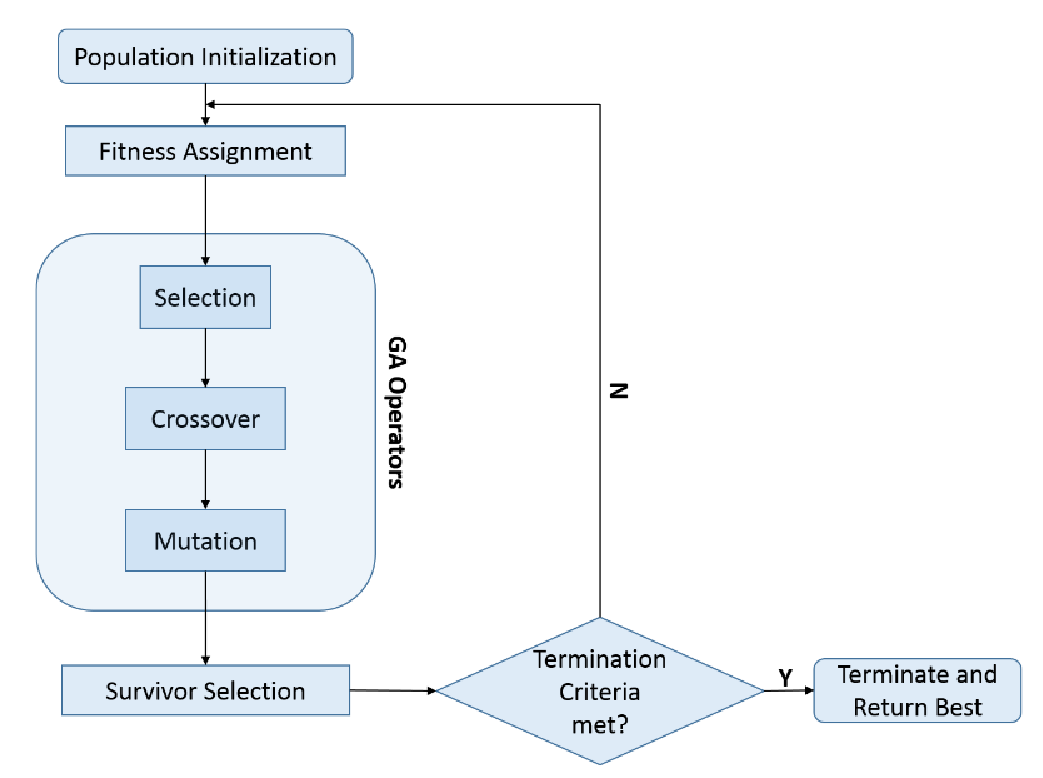}
\end{center}
\vspace{-0.08in}
\caption{ Illustration of Genetic Algorithm.}\label{f2}
\vspace{-0.12in}
\end{figure}
%\end{comment}

{\em 1) Generating the initial population} - The GA is initialized by generating a number of random choices for the parameters of interest $\mathbf{J}$ and $\mathbf{\Phi}$. Each set of parameters is referred to as an ``individual,'' and the parameters are referred to as ``chromosomes.'' In this case, the individuals have $2N$ random chromosomes each, $N$ representing the RIS phase shifts and $N$ accounting for the construction of the permutation matrix. To begin, the $G$ individuals with the highest fitness are used to form the initial population. This initial population is then evolved to the next generation, following the steps below.

{\em 2) Fitness evaluation} - The fitness function quantifies the suitability of each individual in solving the proposed problem. Since our goal in~\eqref{e20} is to minimize the ergodic sum rate of the network, we define the fitness function as follows:
\begin{equation}\label{e21}
\begin{aligned}
\text{fit}(\mathbf{J},\mathbf{\Phi})=\frac{1}{\sum_{k=1}^{K}r^*_k}.
\end{aligned}
\end{equation}

{\em 3) Individual selection} - According to their fitness value, we identify some individuals from the current population as elites and some as parents to generate offspring. First, the $G_e$ individuals with the highest fitness are selected as ``elites,'' and they will be directly passed on to the next generation to preserve their good genes. Then we randomly select 2$G_c$ parents and pair them as $G_c$ couples based on ``roulette wheel'' selection. A section of the wheel is assigned to each possible selection in proportion to their fitness value, and then a random selection is made, similar to how the roulette wheel is spun. Specifically, if $\text{fit}(\mathbf{J}_t,\mathbf{\Phi}^*_t)$ is the fitness of individual $t$ in the population, its probability of being selected is
\begin{equation}\label{e22}
\begin{aligned}
p_t=\frac{\text{fit}(\mathbf{J}_t,\mathbf{\Phi}_t)}{\sum_{t=1}^{G}\text{fit}(\mathbf{J}_t,\mathbf{\Phi}_t)}.
\end{aligned}
\end{equation}
After 2$G_c$ spins of the wheel, we generate a group of parents that will be used for the crossover operation in the next step. Finally, we sample $G_m=G-G_e-G_c$ individuals from the current population with equal probability, including the elites and selected parents, to mutate and generate new individuals. 

{\em 4) Crossover} - The crossover step extracts chromosomes from different parents and recombines them into potentially superior offspring. We will use the 2$G_c$ previously selected parents to perform multi-point crossover and generate $G_c$ offspring. Since the permutation matrix and phase shift matrix variables are of different types, discrete and continuous, respectively, their crossover processes must be implemented differently. The crossover steps include: randomly generating multiple crossover points, selected parents swapping chromosomes at designated crossover points, and generating new individuals. This is straightforward for the entries of $\mathbf{\Phi}$, which are simply swapped between the parents at the randomly selected locations. For the permutation matrix, the crossover operation must result in offspring that satisfy the integer constraints in~\eqref{e20}. For example, suppose the ordered indices $[2,1,3,5,4]$ and $[3,4,1,5,2]$ define the  permutation matrix $\mathbf{J}$ of the two parents. If the randomly chosen crossover point is 2, the ordered indices for $\mathbf{J}$ in the new individuals will be $[2,4,3,5,1]$ and $[3,1,4,5,2]$, where the entry that was swapped out is reinserted in the previous position of the value that was swapped in. Of the couple that remains after all the crossover steps, only one is selected at random as the output of the crossover operation for the next generation. The details of this step are given in Algorithm~2.

{\em Mutation} - Mutation increases the population diversity and increases the likelihood that offspring with better fitness are generated. After the above procedures, we use multi-point uniform mutation to generate additional new individuals. The steps of this method are the following: (1) multiple mutation points are randomly generated, (2) individuals are selected at random to implement chromosome mutation at the designated points, and (3) new individuals are generated. When mutating the permutation matrix and phase shifts to form a new individual, we will select new values with equal probability from their corresponding constraint sets for the mutation points. Pseudo-code for this step is given in Algorithm 2.

\begin{algorithm}[h]
\algsetup{linenosize=\tiny} \scriptsize
    \caption{GA-Based Optimization Scheme}
    \label{alg:1}
    \begin{algorithmic}[1]
        \STATE Define MU-MISO scenario, large-scale fading coefficients, LoS components of the RIS-cascaded channels, and global statistical CSI of the Rayleigh channel components. Initialize population parameters $G^*, G, G_e, G_c, G_m$, number of crossover points $O_1, O_2$ for $\mathbf{J}$ and $\mathbf{\Phi}$, respectively, and number of mutation points $O_3, O_4$ for $\mathbf{J}$ and $\mathbf{\Phi}$, respectively.
        \STATE Initialize a population of individuals of size $G^*$, where each individual $t$ has a randomly generated chromosome $\left\{\mathbf{J}, \mathbf{\Phi}\right\}_t$. From the set of $G^*$ individuals, identify the $G$ with highest fitness as the initial group.
        \FOR{Count $\leq$ Epoch}
            \STATE Calculate the fitness of each individual $t$ as $\text{fit}(\mathbf{\Phi}^*_t)$.
            \STATE Select $G_e$ individuals with highest fitness from current population as elites.
            \STATE Using Algorithm 2, select 2$G_c$ parents based on roulette wheel selection to perform crossover and generate $G_c$ offspring, and select $G_m$ individuals from the current population with equal probability to generate mutated individuals.
            \STATE Combine elites, crossover offspring, and mutated offspring for next generation;
            \STATE Count = Count+1;
        \ENDFOR
        \STATE Output chromosome of individual with highest fitness in the current population.
    \end{algorithmic}
\end{algorithm}
\vspace{-0.15in}
\begin{algorithm}[h]
\algsetup{linenosize=\tiny} \scriptsize
    \caption{Crossover and Mutation Operations}
    \label{alg:2}
    \begin{algorithmic}[1]
        \STATE Crossover: Perform the following steps for each of the $G_c$ couples.
        \FOR{Count $\leq$ $G_c$}
            \STATE Generate $O_1$ distinct random integers from $\left[1, N\right]$ as crossover points for $\mathbf{J}$. For each point, create new parents from the current parents by swapping the entries of $\mathbf{J}$ at the selected point.
            %\STATE Generate $O_2$ distinct random integers from $\left[1, N\right]$ as crossover points for $\mathbf{J}_r$. For each point, create new parents from the current parents by swapping the entries of $\mathbf{J}_r$ at the selected point.
            \STATE Generate $O_2$ distinct random integers from $\left[1, N\right]$ as crossover points for $\mathbf{\Phi}$. For each point, create new parents from the current parents by swapping the elements of $\mathbf{\Phi}$ at the selected point.
            \STATE Randomly select either remaining parent as a new individual for the next population.
            \STATE Count = Count+1;
        \ENDFOR
        \STATE Mutation: Perform the following steps for all $G_m$ individuals selected for mutation.
        \FOR{Count $\leq$ $G_m$}
            \STATE Generate $O_3$ distinct random integers from $\left[1, N\right]$ as mutation points for $\mathbf{J}$. Swap the index of $\mathbf{J}$ at that point with another one chosen at random.
             %\STATE Generate $O_5$ distinct random integers from $\left[1, N\right]$ as mutation points for $\mathbf{J}_r$. Swap the index of $\mathbf{J}_r$ at that point with another one chosen at random.
            \STATE Generate $O_4$ distinct random integers from $\left[1, N\right]$ as mutation points for $\mathbf{\Phi}$. For each point, change the element of $\mathbf{\Phi}$ at that index to another random value on the unit circle.
            \STATE Count = Count+1;
        \ENDFOR
        \STATE Output the new individuals for the next population.
    \end{algorithmic}
\end{algorithm}

\section{Numerical Results}

In this section, we provide numerical results to evaluate the effectiveness of the proposed ND-RIS CRACK in Section II and validate the theoretical analysis in Section III. We consider an MU-MISO system in which the BS and ND-RIS are equipped with uniform linear arrays whose elements are separated by one-half wavelength each. The BS is located at the 3D coordinates (5m, 35m, 20m), and the LUs are randomly distributed in a circular region $S$ centered at (5m, 0m, 1.5m) with a radius of 10m. The ND-RIS is deployed at the location (0m, 30m, 15m). The path-loss exponents of the RIS-User link and BS-User link are set as $\iota_{k,r}=2.5$ and $\iota_{k,b}=3.5$, respectively. Unless otherwise specified, the path-loss exponent of the RIS-BS link is $\iota_{r,b}=2$ and the BS transmit power is $P_k=20$ dBm. When implementing GA, we set $O_i=3$ for all $i$ except $O_2=5$, $G=50$, $G_e=5$, $G_c=35$, $G_m=10$, $G^*=100$ for $\iota_{r,b}=2$, and $G^*=300$ for $\iota_{r,b}=2.5$ or $2.8$. Several different cases for the number of ND-RIS elements and BS antennas are considered, as listed in Table~\ref{tab: SIMULATION PARAMETERS} together with other simulation parameters.

\begin{table}[htp]
    \vspace{-0.04in}
    \centering
    \caption{SIMULATION PARAMETERS}
    \label{tab: SIMULATION PARAMETERS}
    \begin{tabular}{|l|l|c|}\hline
        \textbf{Parameter}& \textbf{Value}\\\hline
        Number of users $K$ & 4 \\\hline
        Number of BS antennas $M$ & [16,32,64,128,256,512] \\\hline
        Number of RIS elements $N$ & [16,32,64,128,256,512] \\\hline
        Population size $G$ & 50 \\\hline
        Path loss $\rho$ & $-20$ dB \\\hline
        Carrier frequency& $28$ GHz \\\hline
        Bandwidth $BW$ & $10$ MHz \\\hline
        Noise power $\sigma^2$ & $ -170+10\rm{log}_{10} (BW)$~dBm \\\hline
        Rician factor $\kappa_{k,r}$ & 1.995 \\\hline
        Rician factor $\kappa_{r,b}$ & [2,4,8,16,32] \\\hline
    \end{tabular}
    \vspace{-0.01in}
\end{table}

\subsection{Benchmark Schemes and Metrics}
We will compare our proposed GA-based ND-RIS CRACK to the following benchmark schemes:
\begin{itemize}
\item[*] \textbf {Without ND-RIS}: The MU-MISO system with no ND-RIS present.
\item[*] \textbf {Random ND-RIS}: The ND-RIS is randomly chosen to implement the reciprocity attack, where $\mathbf{J}$ and $\mathbf{\Phi}$ follow uniform distributions where each possible state is equally likely.
\item[*] \textbf {Heuristic Algorithm 1 (HA1)}: The permutation $\mathbf{J}$ is randomly generated, and based on knowledge of the MU-MISO system's LoS channels, $\mathbf{\Phi}$ is designed to minimize the strength of the LoS BS-RIS-User links $\sum_{k=1}^{K}\Vert \overline{\mathbf{h}}_{r,k}\mathbf{\Phi}^*\overline{\mathbf{H}}_{b,r}\Vert^2$ in each transmission period. When the permutation matrix $\mathbf{J}$ is fixed, the optimization can be solved by the method proposed in~\cite[(41)-(44)]{9737373}.
\item[*] \textbf {Heuristic Algorithm 2 (HA2)}: The phase shifts $\mathbf{\Phi}$ are randomly generated, and based on knowledge of the MU-MISO system's LoS channels, $\mathbf{J} $ is designed to minimize the strength of the LoS BS-RIS-User links $\sum_{k=1}^{K}\Vert \overline{\mathbf{h}}_{r,k}\mathbf{\Phi}^*\overline{\mathbf{H}}_{b,r}\Vert^2$ in each transmission period. When $\mathbf{\Phi}$ is fixed, $\mathbf{J}$ can be optimized using GA.
%\item[*] \mathbf {XXX}. a...b
\end{itemize}

\subsection{Non-Optimized CRACK}

In the following simulations of non-optimized CRACK, we randomly generate 300 ND-RIS realizations, and for each ND-RIS we average the achieved sum rate over 2000 random channel realizations to calculate the average sum ergodic rate.

\begin{comment}
\begin{figure}[htbp]
\centering
\vspace{-0.01in}
\begin{minipage}[t]{0.4\textwidth}
\centering
\includegraphics[width=6cm]{figs/Rayleigh3.eps}
\caption{Rayleigh Channel ($\kappa_{k,r}=0$).}\label{p1}
\end{minipage}
\begin{minipage}[t]{0.4\textwidth}
\centering
\includegraphics[width=6cm]{figs/rician2.eps}
\caption{Rician Channel ($\kappa_{k,r}=1.995$).}\label{p2}
\end{minipage}
\vspace{-0.2in}
\end{figure}
\end{comment}

\begin{comment}

\begin{figure}[ht]
\begin{center}
\vspace{-0.1in}
\includegraphics[width=3.5 in]{figs/rician_rayleigh.eps}
\end{center}
\vspace{-0.1in}
\caption{Sum ergodic rate versus different ND-RIS values (MRT, M=128, N=32).}\label{p1}
\vspace{-0.11in}
\end{figure}

\begin{figure}[ht]
\begin{center}
\vspace{-0.1in}
\includegraphics[width=3.5 in]{figs/rician_rayleigh_ZF1.eps}
\end{center}
\vspace{-0.1in}
\caption{Sum ergodic rate versus different ND-RIS values (ZF, M=128, N=32).}\label{p2}
\vspace{-0.11in}
\end{figure}
\end{comment}

\begin{figure}[t]
\centering
\subfigure[MRT]
{
    \begin{minipage}[b]{.8\linewidth}
        \centering
        \includegraphics[width = 7cm,height = 5.5cm ]{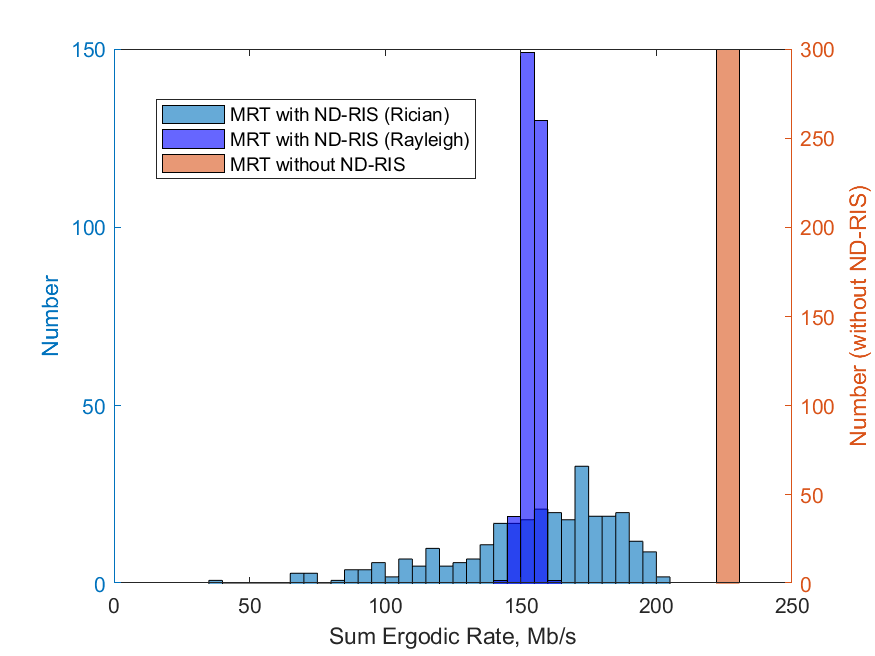}
    \end{minipage}
}
\subfigure[ZF]
{
 	\begin{minipage}[b]{.8\linewidth}
        \centering
        \includegraphics[width =7cm,height = 5.5cm]{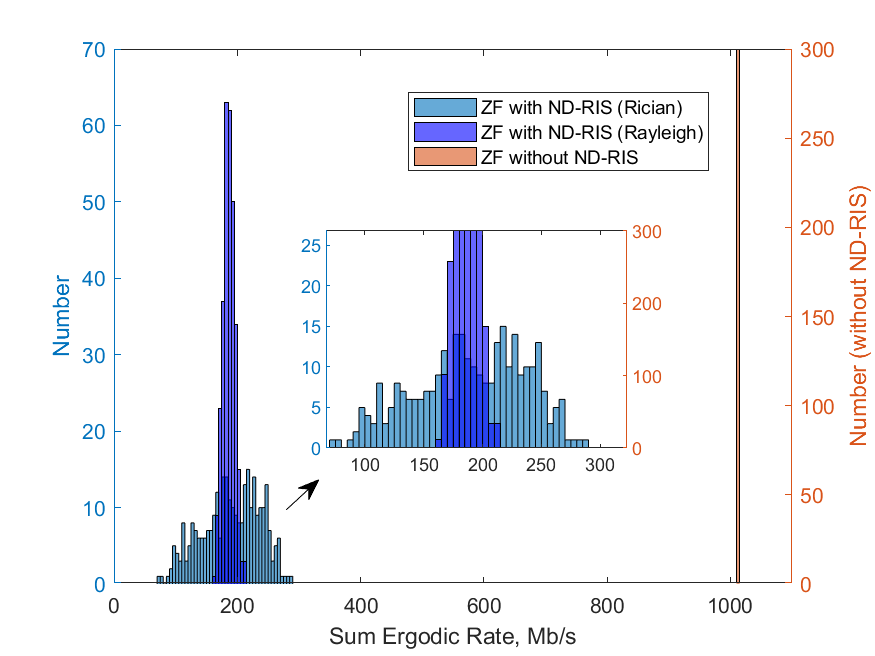}
    \end{minipage}
}
\vspace{-0.08in}
\caption{Sum ergodic rate versus different ND-RIS values for (a) MRT and (b) ZF, where M=128, N=32.}\label{g2}
\vspace{-0.13in}
\end{figure}

\subsubsection{CRACK under Rayleigh and Rician RIS-User channels}
Fig.~\ref{g2} shows the histograms of the sum ergodic rate of the MU-MISO system for Rayleigh and Rician RIS-User channels, assuming the BS employs MRT and ZF beamforming. 
For the case of Rayleigh RIS-User fading, there is little variation in the results
as the ND-RIS varies from trial to trial, and CRACK achieves approximately a $32\%$ decrease in the sum rate compared to the case without the ND-RIS for MRT and approximately a $83\%$ decrease for ZF. For the Rician RIS-User case, the existence of a strong LoS RIS-User link increases the variance of the sum rate under CRACK for both two beamforming schemes, with dramatic degradation in certain cases when the ND-RIS configuration significantly modifies the uplink and downlink LoS path. This motivates use of the GA optimization when knowledge of the LoS path can be exploited, as demonstrated later.

\begin{figure}[t]
\begin{center}
\vspace{-0.1in}
\includegraphics[width=2.8 in]{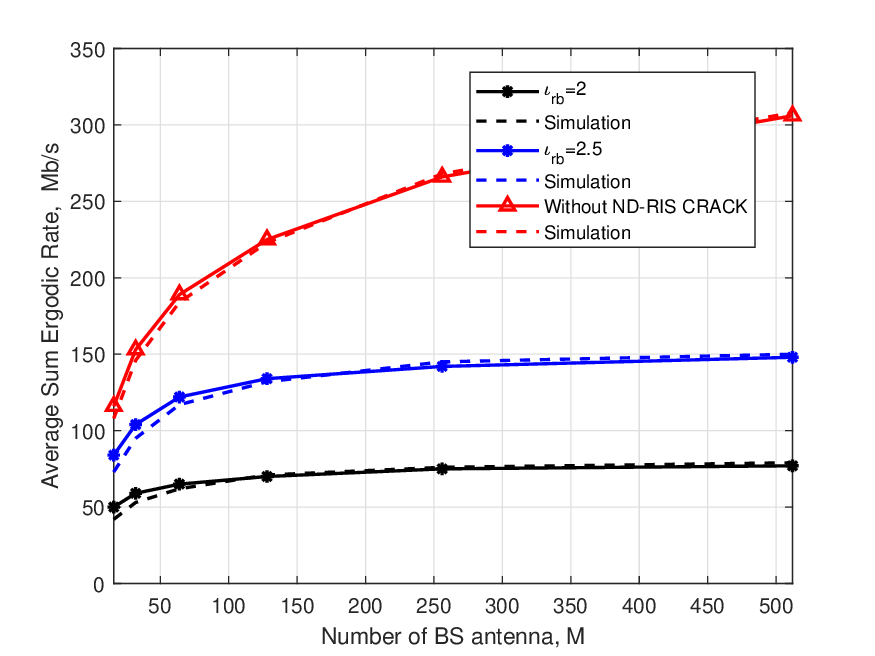}
\end{center}
\vspace{-0.1in}
\caption{Average sum ergodic rate versus different number of BS antennas (MRT, N=128).}\label{p3}
\vspace{-0.11in}
\end{figure}

\subsubsection{MRT beamforming on CRACK}
Fig.~\ref{p3} plots the simulated and predicted average sum ergodic rate as a function of the number of BS antennas and highlights the influence of the BS-RIS channel path loss on CRACK performance, assuming the BS employs MRT beamforming. The simulation results validate the accuracy of our approximate ergodic rate expression, particularly for large values of $M$. When the ND-RIS is present, we see that CRACK cannot be thwarted by continually increasing $M$, as predicted by our analysis which shows that the ergodic sum rate converges to a fixed value for large $M$. In contrast, and as expected, the effectiveness of CRACK relies on the quality of the RIS-BS channel. A higher path loss results in a diminished cascaded RIS channel gain, leading to a weaker CRACK impact. 

\begin{figure}[!t]
\begin{center}
\vspace{-0.1in}
\includegraphics[width=2.8 in]{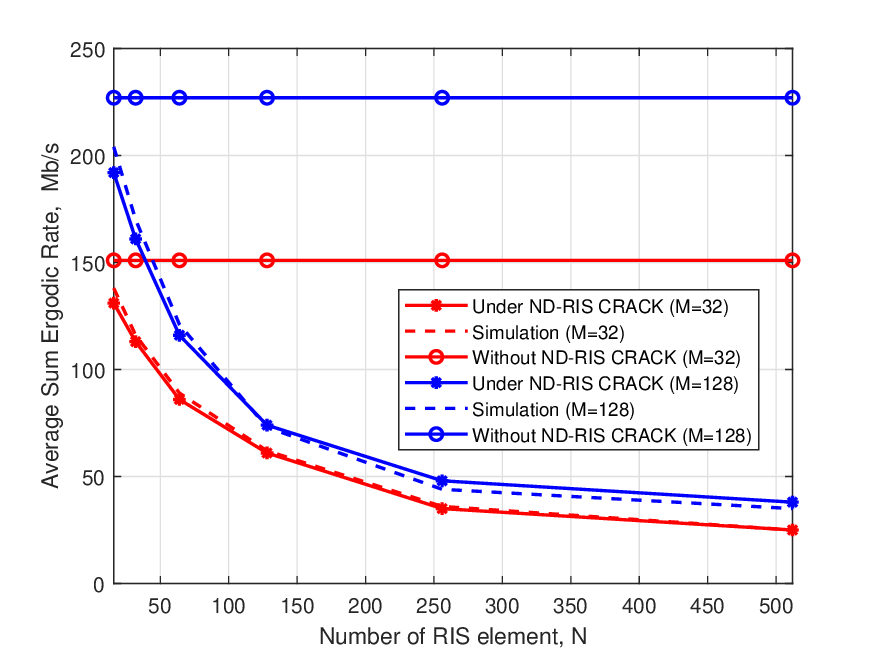}
\end{center}
\vspace{-0.1in}
\caption{Average sum ergodic rate versus ND-RIS size (MRT).}\label{p4}
\vspace{-0.05in}
\end{figure}

In Fig.~\ref{p4}, we show the average sum ergodic rate as a function of the number of ND-RIS elements $N$ for $M=32$ and $M=128$ BS antennas while also comparing the impact of the BS antenna array size on the attack performance, assuming the BS employs MRT beamforming. As $N$ increases, the sum ergodic rate initially experiences a steep decline before stabilizing at an extremely low level. This observation highlights the significant detrimental effect of CRACK on the MU-MISO system, where a large degradation can be achieved without relying on CSI for the ND-RIS design. In particular, we observe a reduction of nearly $82\%$ in the sum ergodic rate when $N=512$ and $M=128$.

\begin{figure}[!t]
\begin{center}
\vspace{-0.1in}
\hspace*{-0.5cm}\includegraphics[width=4in]{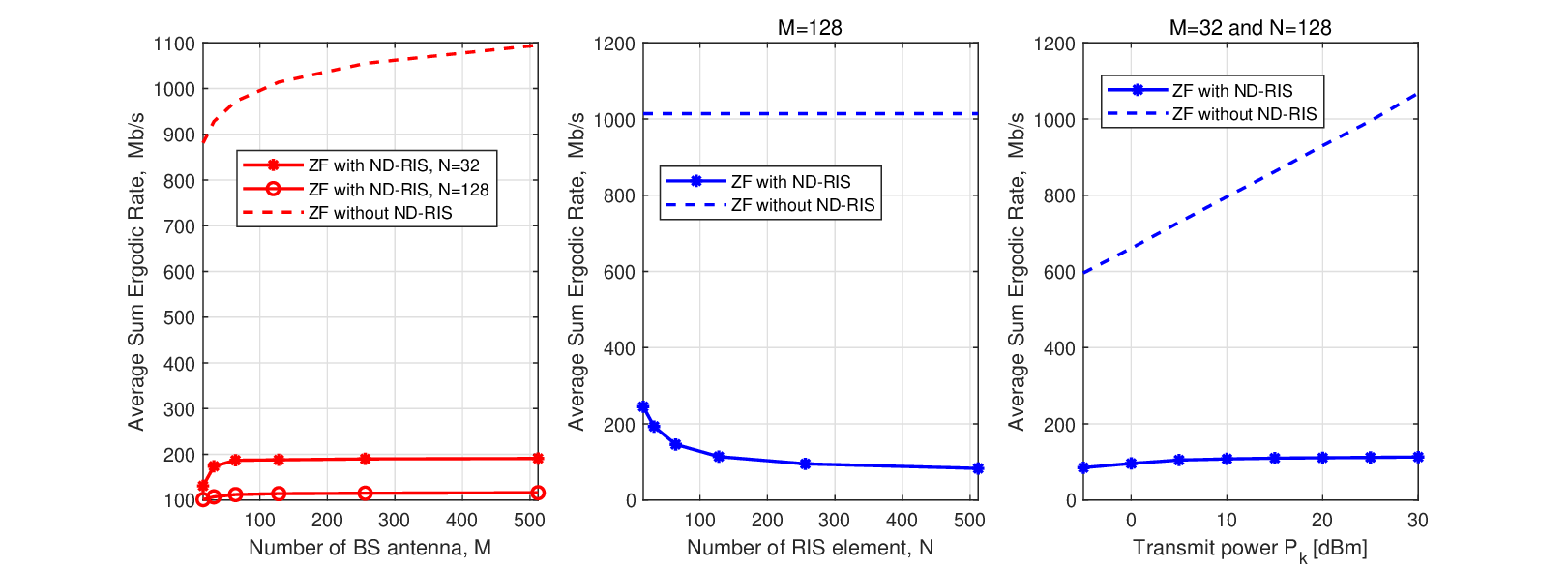}
\end{center}
\vspace{-0.2in}
\caption{Average ergodic sum rate of MU-MISO system with ZF precoding.}\label{p8}
\vspace{-0.13in}
\end{figure}

\subsubsection{ZF beamforming on CRACK}
As illustrated in the first example, the proposed CRACK approach degrades the MU-MISO system throughput for other precoders besides MRT. In Fig.~\ref{p8}, we present empirical results for the case of ZF precoding, showing the sum ergodic rate versus the number of BS antennas, RIS elements, and BS transmit power, for the case where the ND-RIS configuration is randomly determined. Compared with the results for MRT in the previous section, we see that CRACK yields a significantly larger degradation in performance for ZF precoding, since ZF ideally eliminates all multiuser interference. Similar trends are observed as in the previous cases, with additional BS antennas not providing any relief to the attack, and an increasing RIS size providing significant benefits. The rightmost figure further illustrates that increasing the BS transmit power cannot be used as a mechanism to overcome CRACK, as would be the case for other types of physical layer attacks.

\begin{comment}
\begin{figure}[ht]
\begin{center}
\vspace{-0.1in}
\includegraphics[width=3 in]{figs/ZF2.eps}
\end{center}
\vspace{-0.2in}
\caption{Sum ergodic rate of MU-MISO Communication System adopting ZF scheme (M=128, N=32)}\label{p9}
\vspace{-0.13in}
\end{figure}

Fig.~\ref{p9} illustrates the effects of different ND-RIS values (randomly generated) on the sum ergodic rate when BS adopts ZF precoding scheme. A randomly changing ND-RIS with a small size (N=32) has realized a nearly $81\%$ rate decline. However, the frequent changes in the channel environment caused by ND-RIS are easily detected by BS. It may stop the data transmission or even manhunt the attacking device. Despite the lack of explicit ergodic rate expression, malicious attackers can utilize the collected statistical CSI to simulate the entire attack process, and then select the ND-RIS value leading to the lowest ergodic rate as the optimal design. Like the point highlighted in Fig.~\ref{p9}, it causes nearly $93\%$ degradation of the system performance, which has almost destroyed the entire communication process.
\end{comment}

\subsubsection{$1$-bit ND-RIS CRACK}
In this example, we demonstrate that the success of CRACK does not depend on continuous phase control of the ND-RIS elements. To show this, we assume an ND-RIS whose phase shifts can be in one of two states, $\theta \in \left\{0,\pi\right\}$, with one bit of control. Fig.~\ref{p10} illustrates the sum ergodic rate for both MRT and ZF precoding at the BS versus the number of ND-RIS elements, comparing the degradation of a 1-bit ND-RIS and a normal ND-RIS with the system without the ND-RIS. We see that the results are very similar to those for the case with full phase control.
%Continuous phase RIS is a kind of ideal assumption, and most RIS in practical scenarios are discrete phase shift control. Large RIS with precise phase shift control will introduce manufacturing difficulties and hardware costs. The 1-bit RIS is easy to manufacture and operate and has helped communications systems achieve significant performance enhancement~\cite{8485924,5765450}. To this end, we investigate the impact of CRAS based on 1-bit ND-RIS on MU-MISO system performance. 

\begin{figure}[t]
\begin{center}
\vspace{-0.05in}
\includegraphics[width=2.8 in]{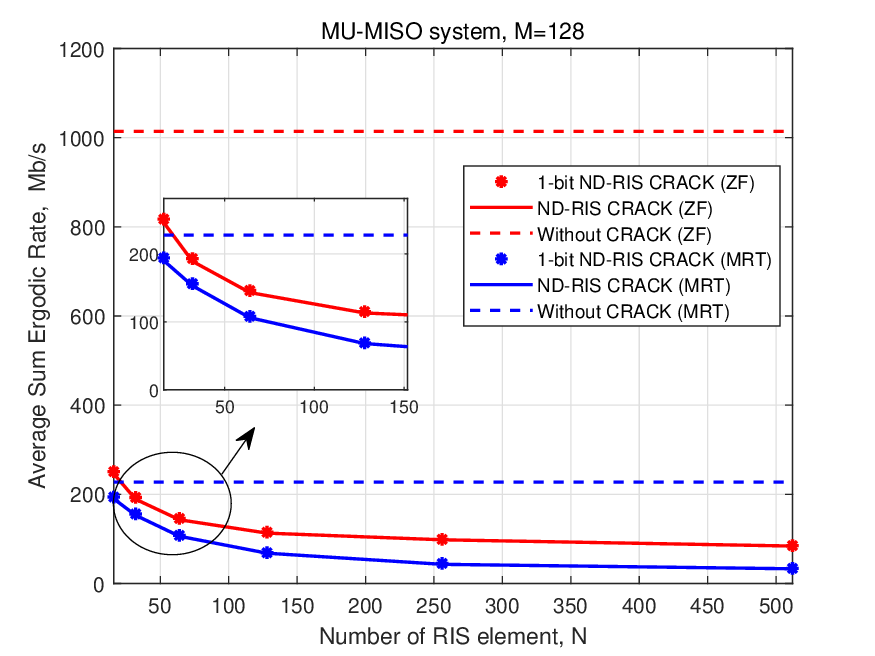}
\end{center}
\vspace{-0.2in}
\caption{Average ergodic sum rate under 1-bit ND-RIS CRACK with M=128.}\label{p10}
\vspace{-0.01in}
\end{figure}

\subsubsection{ND-RIS Deployment Strategy}

\begin{figure}[!t]
\begin{center}
\vspace{-0.01in}
\includegraphics[width=2.8 in]{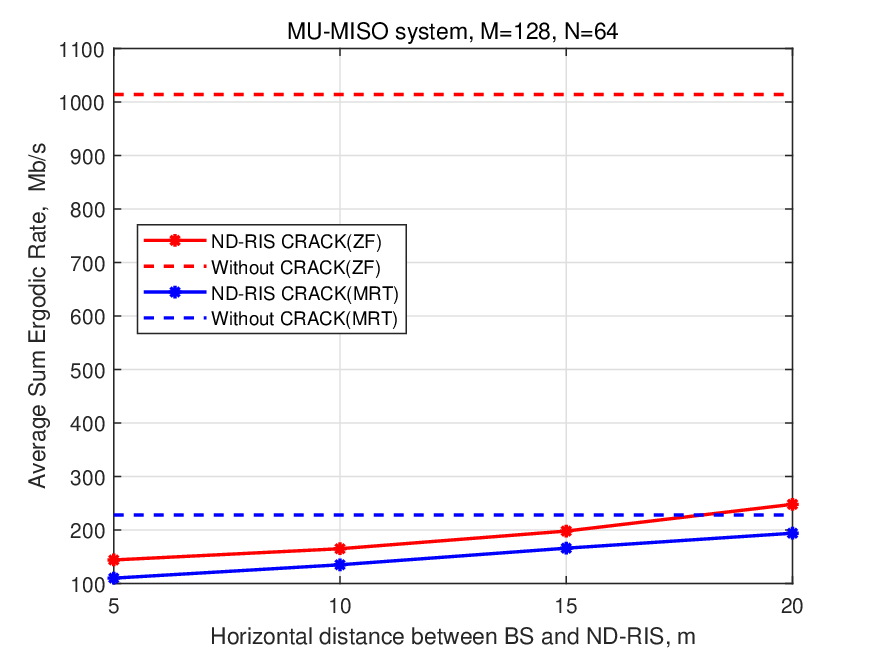}
\end{center}
\vspace{-0.1in}
\caption{Influence of RIS-BS distance on ergodic sum rate under CRACK. }\label{p11}
\vspace{-0.13in}
\end{figure}

This example demonstrates that the ND-RIS will have a more significant influence if it is deployed closer to BS than to the users. To show this, we place the ND-RIS nearly on the line between the BS and the users and vary the distance between the ND-RIS and BS. Fig.~\ref{p11} indicates that as the distance increases and the ND-RIS approaches the users, the sum ergodic rate of the MU-MISO system increases for both ZF and MRT.

\subsubsection{Influence of RIS-BS Rician Factor}

\begin{figure}[ht]
\begin{center}
\vspace{-0.01in}
\includegraphics[width=2.8 in]{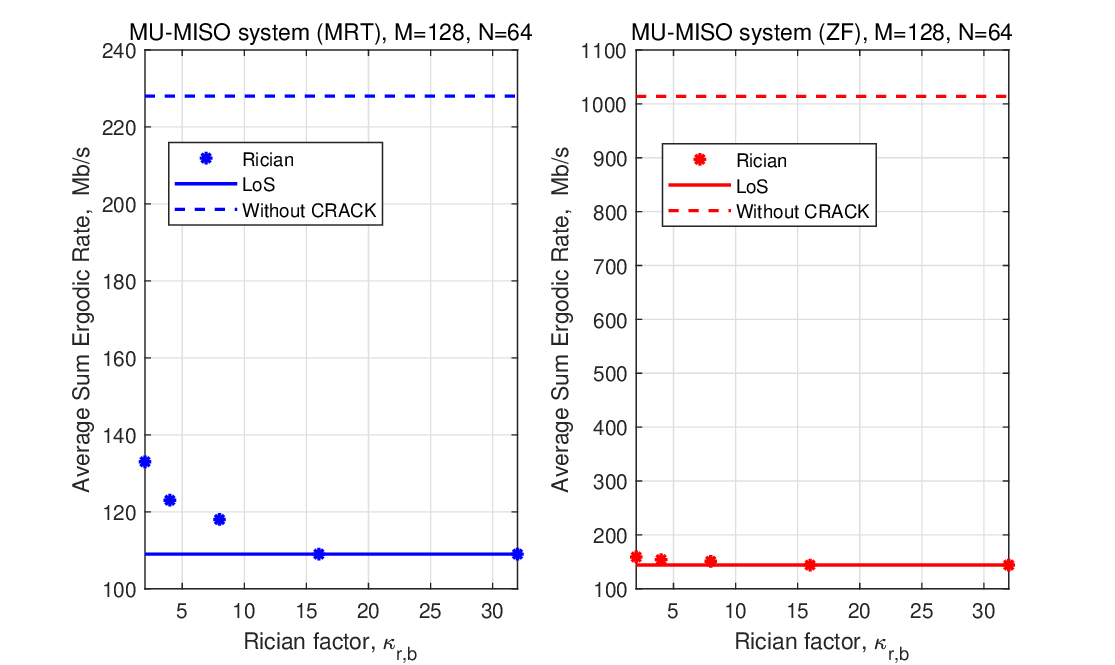}
\end{center}
\vspace{-0.2in}
\caption{Influence of RIS-BS Rician factor on ergodic sum rate under CRACK. }\label{p12}
\vspace{-0.01in}
\end{figure}

Our performance analysis in Section III assumed an LoS-only RIS-BS channel. In this example, we verify that CRACK is nonetheless effective in the more general case of a Rician channel with gain $\kappa_{r,b}$. Fig.~\ref{p12} shows that as the strength of the LoS path increases, the average sum ergodic rate of the MU-MISO system under CRACK will decrease for both ZF and MRT, with a larger impact for MRT. When $\kappa_{r,b} \ge 16$, the Rician and LoS channel models yield very similar results, and in such cases the NLoS component can be ignored. However, even for small $\kappa_{r,b}$, CRACK provides a significant performance degradation.

\subsection{GA-Based CRACK Optimization}

\begin{comment}
\begin{figure}[ht]
\begin{center}
\vspace{-0.1in}
\includegraphics[width=3 in]{figs/GA1.eps}
\end{center}
\vspace{-0.2in}
\caption{Impact of BS-RIS path loss exponent $\iota_{r,b}$ on GA based optimization (M=128 and N=32).}\label{p7}
\vspace{-0.13in}
\end{figure}
\end{comment}

\begin{figure}[!t]
\centering
\vspace{-0.1in}
\begin{minipage}[t]{0.40\textwidth}
\centering
\includegraphics[width=7cm]{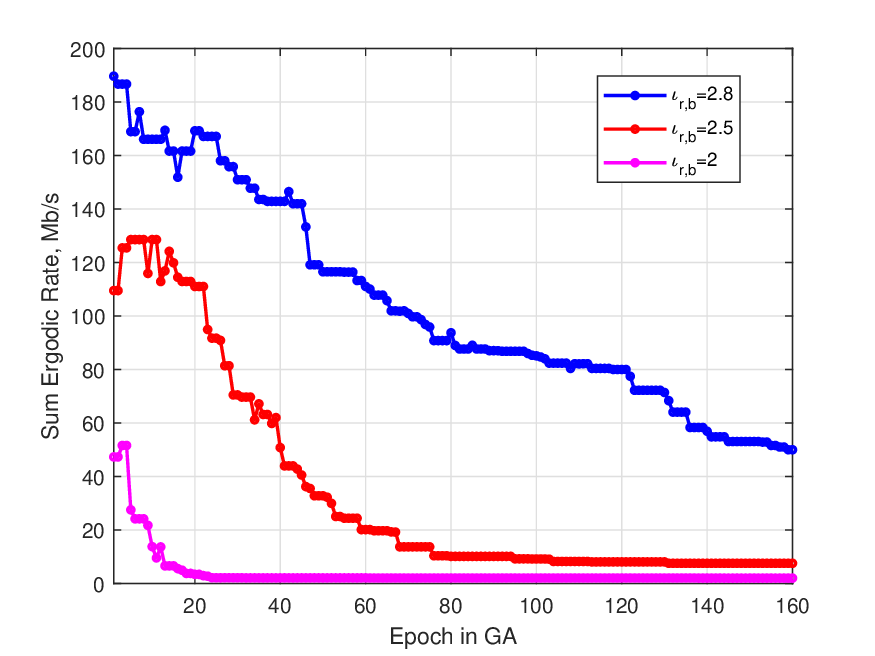}
\vspace{-0.1in}
\caption{Impact of BS-RIS path loss exponent $\iota_{r,b}$ on GA based optimization (M=128 and N=32).}\label{p5}
\end{minipage}
\begin{minipage}[t]{0.40\textwidth}
\centering
\includegraphics[width=7cm]{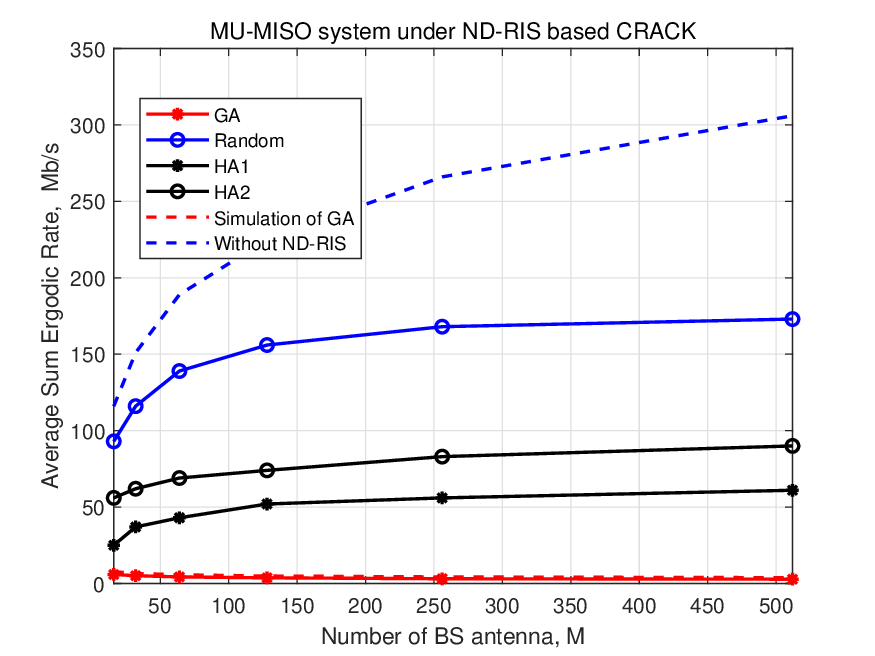}
\vspace{-0.1in}
\caption{GA performance for ND-RIS CRACK (N=32).}\label{p6}
\end{minipage}
\vspace{-0.1in}
\end{figure}

\begin{figure}[t]
\begin{center}
\vspace{-0.1in}
\includegraphics[width=2.8 in]{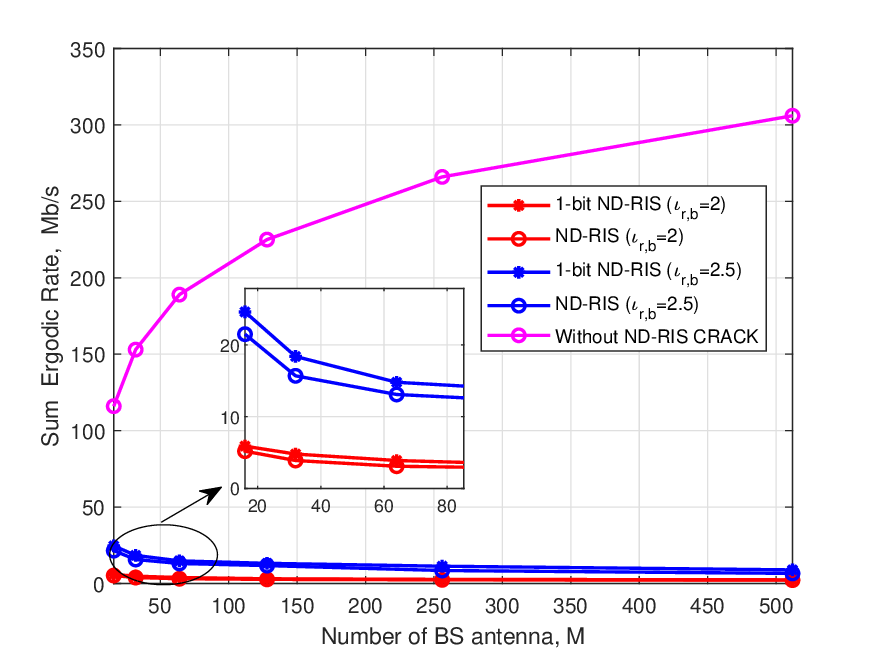}
\end{center}
\vspace{-0.2in}
\caption{Performance of CRACK with 1-bit and full resolution ND-RIS.}\label{p14}
\vspace{-0.13in}
\end{figure}

Here we demonstrate that when a malicious attacker has the opportunity to obtain the statistical CSI and LoS channel components of the MU-MISO system, the ND-RIS can be optimized based on (\ref{e17}) and (\ref{e18}) to minimize the sum ergodic rate of the system. In Fig.~\ref{p5}, we show the convergence of the proposed GA-based optimization for different BS-RIS path-loss exponents. It is observed that in each case, at most 160 epochs are required to achieve reasonable convergence, with some variation due to the influence of the initial selected population. The results clearly indicate that a lower path loss and hence a stronger BS-RIS channel not only yields a stronger impact for CRACK, but also much faster convergence for the algorithm. Next, we use the GA approach to optimize CRACK and compare it with other benchmark schemes. Fig.~\ref{p6} illustrates the average sum ergodic rate versus the number of BS antennas. As $M$ increases, the average sum ergodic rate initially grows for the other benchmark approaches and then stabilizes at a value that is significantly higher than that achieved using the GA optimization, which is near zero. When only the LoS information is available, the heuristic algorithms HA1 and HA2 provide a significant degradation in MU-MISO performance compared to randomly selecting the ND-RIS phases, with HA2 (fixed random phase, optimized NS-RIS connections) providing the most significant reduction, at nearly 81\% compared to the case without the ND-RIS.

In this example we use GA to optimize the design of both a 1-bit ND-RIS and an ND-RIS with full resolution phase control and compare their CRACK performance. Fig.~\ref{p14} shows the resulting ergodic sum rate versus the number of BS antenna and for two different values of the BS-RIS Rician factor $\kappa_{r,b}$. We see again that there is very little difference performance when a 1-bit rather than a full resolution ND-RIS is deployed, especially when $\kappa_{r,b}$ is small. As $M$ increases, the effect of the attack gradually increases for both types of ND-RIS, although there is not much increase in link degradation for $M$ greater than about 50 antennas. Similar to the non-optimized results in Fig.~\ref{p8}, the asymptotic behavior described in Corollary~1 holds for relatively small values of $M$. The effect of the BS-RIS Rician factor is also clearly apparent, where an increase in $\kappa_{r,b}$ from 2 to 2.5 results in an increase in sum ergodic rate by factor of 2-4 depending on the value of $M$.

\section{Conclusions}
We have proposed a novel ND-RIS-based channel reciprocity attack or CRACK approach that can be implemented to deteriorate the performance of time-division duplex communication systems. By passively eliminating the assumed reciprocity of the legitimate channel, CRACK introduces unanticipated multi-user interference into the system under attack that cannot be countered using channel estimation or by simply increasing the number of BS antennas or transmit power. In addition, the attack can be launched without any CSI about the BS-user links or any links involving the ND-RIS, it does not require creating rapid time variations in the channel, and it does not require synchronization with the training or data phases of the communication system. However, if statistical CSI is available together with information about the LoS links, then the performance degradation due to CRACK can be considerably enhanced. We developed a specialized GA-based algorithm to implement the optimization of the ND-RIS element phases and inter-element connections, and demonstrated via simulation the resulting gains. We also showed that the phase of the ND-RIS elements need only be selected with one-bit of precision in order to achieve considerable performance loss. We further derived a closed-form approximate expression for the sum ergodic rate of an MU-MISO system under CRACK, and demonstrated that the derived expressions closely match the resulting system performance. 

\begin{comment}
\begin{appendices}

\section{}
MMM
\section{}
xxx
\end{appendices}
\end{comment}

\appendix
\renewcommand{\appendixname}{Appendix ~\Alph{section}}

To begin with, we present some definitions that will be used in later derivations. According to the definition of the 
Rician channel~(\ref{e9}) and the
LoS channel (\ref{e10}), we can rewrite the RIS-based cascaded channels $\mathbf{g}_{k}$, $\mathbf{q}_{k}$, and $\mathbf{q}_{i}$ as
\begin{equation}\label{e23}
\begin{aligned}
\mathbf{g}_{k}=\sqrt{\frac{\alpha_{k,r}\alpha_{r,b}}{(1+\kappa_{k,r})}}\left( \underbrace{\sqrt{\kappa_{k,r}}\overline{\mathbf{h}}_{r,k}\mathbf{\Phi}^*\overline{\mathbf{H}}_{b,r}}_{\mathbf{g}^1_{k}}
+\underbrace{\widetilde{\mathbf{h}}_{r,k}\mathbf{\Phi}^*\overline{\mathbf{H}}_{b,r}}_{\mathbf{g}^2_{k}}
\right)
\end{aligned}
\end{equation}
\begin{equation}\label{e24}
\begin{aligned}
\mathbf{q}_{k}=\sqrt{\frac{\alpha_{k,r}\alpha_{r,b}}{(1+\kappa_{k,r})}}\left( \underbrace{\sqrt{\kappa_{k,r}}\overline{\mathbf{H}}^H_{b,r}\widetilde{\mathbf{\Phi}}\overline{\mathbf{h}}^H_{r,k}}_{\mathbf{q}^1_{k}}
+\underbrace{\overline{\mathbf{H}}^H_{b,r}\widetilde{\mathbf{\Phi}}\widetilde{\mathbf{h}}^H_{r,k}}_{\mathbf{q}^2_{k}}
\right)
\end{aligned}
\end{equation}
\begin{equation}\label{e25}
\begin{aligned}
\mathbf{q}_{i}=\sqrt{\frac{\alpha_{i,r}\alpha_{r,b}}{(1+\kappa_{i,r})}}\left( \underbrace{\sqrt{\kappa_{i,r}}\overline{\mathbf{H}}^H_{b,r}\widetilde{\mathbf{\Phi}}\overline{\mathbf{h}}^H_{r,i}}_{\mathbf{q}^1_{i}}
+\underbrace{\overline{\mathbf{H}}^H_{b,r}\widetilde{\mathbf{\Phi}}\widetilde{\mathbf{h}}^H_{r,i}}_{\mathbf{q}^2_{i}}
\right),
\end{aligned}
\end{equation}
respectively, where $\widetilde{\mathbf{H}}_{r,b}$, $\widetilde{\mathbf{h}}_{k,r}$, and $\widetilde{\mathbf{h}}_{i,r}$ are
independent of each other. In addition, we define
\begin{equation}\label{e26}
f_k(\mathbf{\Phi}^*)\triangleq \overline{\mathbf{h}}_{r,k}\mathbf{\Phi}^*\mathbf{a}_N
\end{equation}
\begin{equation}\label{e27}
f^*_k(\mathbf{\Phi}^*)\triangleq \mathbf{a}^H_N \widetilde{\mathbf{\Phi}}\overline{\mathbf{h}}^H_{r,k}
\end{equation}
\begin{equation}\label{e28}
f_i(\mathbf{\Phi}^*)\triangleq\mathbf{a}^H_N \widetilde{\mathbf{\Phi}}\overline{\mathbf{h}}^H_{r,i},
\end{equation}
for the convenience of the subsequent analysis. Next, we will derive each expectation in~(\ref{e17}) and (\ref{e18}).

\subsection{Derivation of $\mathbb{E}\left[|\mathbf{g}_{k}\mathbf{q}_{i}|^2\right]$}
Before the proof, we first use the following result, which holds when each element in $\widetilde{\mathbf{H}} \in \mathbb{C}^{M \times N}$ follows an i.i.d. complex Gaussian distribution with zero mean and unit variance, and $\mathbf{A} \in \mathbb{C}^{N\times M}$ is an arbitrary deterministic matrix~\cite{9743440}:
\begin{equation}\label{e29}
\mathbb{E}\left\{\rm{Re}\left\{\widetilde{\mathbf{H}}\mathbf{A}\widetilde{\mathbf{H}}\right\}\right\}=\mathbf{0}.
\end{equation}
Given~(\ref{e23}) and~(\ref{e25}), we can ignore terms with zero expectation and terms with zero real part based on~(\ref{e29}), so that
\begin{equation}\label{e30}
\begin{aligned}
\mathbb{E}|\mathbf{g}_{k}\mathbf{q}_{i}|^2&=\frac{\alpha^2_{r,b}\alpha_{k,r}\alpha_{i,r}}{(1+\kappa_{k,r})(1+\kappa_{i,r})}\mathbb{E}\left\{\left| \sum_{u=1}^{2}\sum_{v=1}^{2}\mathbf{g}^{u}_{k}\mathbf{q}^{v}_{i}   \right|^2\right\}\\&=\frac{\alpha^2_{r,b}\alpha_{k,r}\alpha_{i,r}}{(1+\kappa_{k,r})(1+\kappa_{i,r})}\mathbb{E}\left\{ \sum_{u=1}^{2}\sum_{v=1}^{2}\left|\mathbf{g}^{u}_{k}\mathbf{q}^{v}_{i}   \right|^2\right\},
\end{aligned}
\end{equation}
Next, we will calculate the terms in (\ref{e30}) one by one.
\subsubsection{$\mathbb{E}\left| \mathbf{g}^{u}_{k}\mathbf{q}^{v}_{i}   \right|^2$, $u=1,v\leq 2$}
When $v=1$, we have
\begin{equation}\label{e31}
\begin{aligned}
\mathbb{E}\left| \mathbf{g}^{1}_{k}\mathbf{q}^{1}_{i}   \right|^2
&=\mathbb{E}\left\{\kappa_{k,r}\kappa_{i,r}\left| \overline{\mathbf{h}}_{r,k}\mathbf{\Phi}^*\overline{\mathbf{H}}_{b,r} \overline{\mathbf{H}}^H_{b,r}\widetilde{\mathbf{\Phi}}\overline{\mathbf{h}}^H_{r,i}  \right|^2\right\}\\
&=\kappa_{k,r}\kappa_{i,r}M^2\left|f_k(\mathbf{\Phi}^*)\right|^2 \left|f_i(\mathbf{\Phi}^*)\right|^2,
\end{aligned}
\end{equation}
When $v=2$, 
\begin{equation}\label{e32}
\begin{aligned}
\mathbb{E}\left| \mathbf{g}^{1}_{k}\mathbf{q}^{2}_{i}   \right|^2
&=\mathbb{E}\left\{\kappa_{k,r}\left| \overline{\mathbf{h}}_{r,k}\mathbf{\Phi}^*\overline{\mathbf{H}}_{b,r} \overline{\mathbf{H}}^H_{b,r}\widetilde{\mathbf{\Phi}}\widetilde{\mathbf{h}}^H_{r,i}  \right|^2\right\}\\
&=\kappa_{k,r}M^2 N\left|f_k(\mathbf{\Phi}^*)\right|^2,
\end{aligned}
\end{equation}

\begin{comment}
$\mathbb{E}\left\{\left| \mathbf{g}^{1}_{k}\mathbf{q}^{1}_{i}   \right|^2\right\}=\mathbb{E}\left\{\kappa_{k,r}\kappa_{i,r}\left| \overline{\mathbf{h}}_{r,k}\mathbf{\Phi}^*\overline{\mathbf{H}}_{b,r} \overline{\mathbf{H}}^H_{b,r}\widetilde{\mathbf{\Phi}}\overline{\mathbf{h}}^H_{r,i}  \right|^2\right\}=\kappa_{k,r}\kappa_{i,r}M^2\left|f_k(\mathbf{\Phi}^*)\right|^2 \left|f_i(\mathbf{\Phi}^*)\right|^2$,\\
$\mathbb{E}\left\{\left| \mathbf{g}^{1}_{k}\mathbf{q}^{2}_{i}   \right|^2\right\}=\mathbb{E}\left\{\kappa_{k,r}\left| \overline{\mathbf{h}}_{r,k}\mathbf{\Phi}^*\overline{\mathbf{H}}_{b,r} \overline{\mathbf{H}}^H_{b,r}\widetilde{\mathbf{\Phi}}\widetilde{\mathbf{h}}^H_{r,i}  \right|^2\right\}=\kappa_{k,r}M^2 N\left|f_k(\mathbf{\Phi}^*)\right|^2$,\\
\end{comment}

\subsubsection{$\mathbb{E}\left| \mathbf{g}^{u}_{k}\mathbf{q}^{v}_{i}   \right|^2$, $u=2,v\leq 2$}
When $v=1$, we have
\begin{equation}\label{e33}
\begin{aligned}
\mathbb{E}\left| \mathbf{g}^{2}_{k}\mathbf{q}^{1}_{i}   \right|^2
&=\mathbb{E}\left\{\kappa_{i,r}\left| \widetilde{\mathbf{h}}_{r,k}\mathbf{\Phi}^*\overline{\mathbf{H}}_{b,r} \overline{\mathbf{H}}^H_{b,r}\widetilde{\mathbf{\Phi}}\overline{\mathbf{h}}^H_{r,i}  \right|^2\right\}\\
&=\kappa_{i,r}M^2N \left|f_i(\mathbf{\Phi}^*)\right|^2,
\end{aligned}
\end{equation}
When $v=2$, 
\begin{equation}\label{e34}
\begin{aligned}
\mathbb{E}\left| \mathbf{g}^{2}_{k}\mathbf{q}^{2}_{i}   \right|^2
&=\mathbb{E}\left\{\left| \widetilde{\mathbf{h}}_{r,k}\mathbf{\Phi}^*\overline{\mathbf{H}}_{b,r} \overline{\mathbf{H}}^H_{b,r}\widetilde{\mathbf{\Phi}}\widetilde{\mathbf{h}}^H_{r,i}  \right|^2\right\}\\
&=M^2 {\rm{Tr}}\left\{\Psi_2\Psi^H_2\right\}.
\end{aligned}
\end{equation}
Substituting (\ref{e31})-(\ref{e34}) into  (\ref{e30}), we can obtain $\mathbb{E}|\mathbf{g}_{k}\mathbf{q}_{i}|^2$.

\begin{comment}
$\mathbb{E}\left\{\left| \mathbf{g}^{2}_{k}\mathbf{q}^{1}_{i}   \right|^2\right\}=\mathbb{E}\left\{\kappa_{i,r}\left| \widetilde{\mathbf{h}}_{r,k}\mathbf{\Phi}^*\overline{\mathbf{H}}_{b,r} \overline{\mathbf{H}}^H_{b,r}\widetilde{\mathbf{\Phi}}\overline{\mathbf{h}}^H_{r,i}  \right|^2\right\}=\kappa_{i,r}M^2N \left|f_i(\mathbf{\Phi}^*)\right|^2$,\\
$\mathbb{E}\left\{\left| \mathbf{g}^{2}_{k}\mathbf{q}^{2}_{i}   \right|^2\right\}=\mathbb{E}\left\{\left| \widetilde{\mathbf{h}}_{r,k}\mathbf{\Phi}^*\overline{\mathbf{H}}_{b,r} \overline{\mathbf{H}}^H_{b,r}\widetilde{\mathbf{\Phi}}\widetilde{\mathbf{h}}^H_{r,i}  \right|^2\right\}=M^2 {\rm{Tr}}\left\{\Psi_2\Psi^H_2\right\}$,\\
\end{comment}

\subsection{Derivation of $\mathbb{E}|\mathbf{g}_{k}\mathbf{h}^H_{b,i}|^2$, $\mathbb{E}|\mathbf{h}_{b,k}\mathbf{q}_{i}|^2$, and $\mathbb{E}|\mathbf{h}_{b,k}\mathbf{h}^H_{b,i}|^2$ }
Based on the independence of the channels, we have $\mathbb{E}\left[\mathbf{g}^u_{k}(\mathbf{g}^{v}_{k})^H\right]=0$ for $\forall u \neq v$. Then 
\begin{equation}\label{e35}
\begin{aligned}
\mathbb{E}|\mathbf{g}_{k}\mathbf{h}^H_{b,i}|^2&=\alpha_{i,b}\mathbb{E}\left[\mathbf{g}_{k}\mathbf{g}^H_{k} \right]\\
&=\frac{\alpha_{i,b}\alpha_{r,b}\alpha_{k,r}}{(1+\kappa_{k,r})}\sum_{u=1}^{2}\mathbb{E} \Vert\mathbf{g}^u_{k}\Vert^2\\
&=\frac{\alpha_{i,b}\alpha_{r,b}\alpha_{k,r}}{(1+\kappa_{k,r})}\left(\kappa_{k,r}\Vert \overline{\mathbf{h}}_{r,k}\mathbf{\Phi}^*\overline{\mathbf{H}}_{b,r}\Vert^2 \right. \\ 
 & \qquad\left.+ \mathbb{E}\left[\Vert \widetilde{\mathbf{h}}_{r,k}\mathbf{\Phi}^*\overline{\mathbf{H}}_{b,r}\Vert^2\right] \right)\\
&\overset{(b)}{=}\frac{\alpha_{i,b}\alpha_{r,b}\alpha_{k,r}\left(\kappa_{k,r} M\left|f_k(\mathbf{\Phi}^*)\right|^2+MN\right)}{(1+\kappa_{k,r})},
\end{aligned}
\end{equation}
where $(b)$ utilizes the following results:
\begin{eqnarray}
\Vert \overline{\mathbf{h}}_{r,k}\mathbf{\Phi}^*\overline{\mathbf{H}}_{b,r}\Vert^2 & = & M\left|f_k(\mathbf{\Phi}^*)\right|^2 \\
\mathbb{E}\; \Vert \widetilde{\mathbf{h}}_{r,k}\mathbf{\Phi}^*\overline{\mathbf{H}}_{b,r}\Vert^2 & = & MN \; .
\end{eqnarray}
Similarly, $\mathbb{E}|\mathbf{h}_{b,k}\mathbf{q}_{i}|^2 $ can be represented as
\begin{equation}\label{e36}
\begin{aligned}
\mathbb{E}|\mathbf{h}_{b,k}\mathbf{q}_{i}|^2 &=\frac{\alpha_{k,b}\alpha_{r,b}\alpha_{i,r}}{(1+\kappa_{i,r})}\sum_{v=1}^{2} \mathbb{E}\Vert\mathbf{q}^v_{i}\Vert^2\\
&=\frac{\alpha_{k,b}\alpha_{r,b}\alpha_{i,r}}{(1+\kappa_{i,r})}\left(\kappa_{i,r}\Vert \overline{\mathbf{H}}^H_{b,r}\widetilde{\mathbf{\Phi}}\overline{\mathbf{h}}^H_{r,i}\Vert^2 \right. \\ 
 & \qquad\left.+ \mathbb{E}\left[\Vert \overline{\mathbf{H}}^H_{b,r}\widetilde{\mathbf{\Phi}}\widetilde{\mathbf{h}}^H_{r,i}\Vert^2\right] \right)\\
&=\frac{\alpha_{k,b}\alpha_{r,b}\alpha_{i,r}\left(\kappa_{i,r} M\left|f_i(\mathbf{\Phi}^*)\right|^2+MN\right)}{(1+\kappa_{i,r})}.
\end{aligned}
\end{equation}
Since $\mathbf{h}_{b,k}$ and $\mathbf{h}_{b,i}$ are independent, we have
\begin{equation}\label{e37}
\begin{aligned}
\mathbb{E}|\mathbf{h}_{b,k}\mathbf{h}^H_{b,i}|^2
&=\mathbb{E}\left[\mathbf{h}_{b,k}\mathbf{h}^H_{b,i}\mathbf{h}_{b,i}\mathbf{h}^H_{b,k} \right] \\
&=\mathbb{E}\left[\alpha_{i,b}\mathbf{h}_{b,k}\mathbf{h}^H_{b,k} \right] =\alpha_{k,b}\alpha_{i,b}M.
\end{aligned}
\end{equation}

\subsection{Derivation of $\mathbb{E}\left[P_k|(\mathbf{g}_{k}+\mathbf{h}_{b,k})(\mathbf{q}_{k}+\mathbf{h}^H_{b,k})|^2\right]$}
We begin with the following sequence of derivations in (\ref{e38}),
\begin{figure*}[h]
\begin{equation}\label{e38}
\begin{aligned}
&\mathbb{E}\left[P_k|(\mathbf{g}_{k}+\mathbf{h}_{b,k})(\mathbf{q}_{k}+\mathbf{h}^H_{b,k})|^2\right]
=P_k\mathbb{E}\left[|\mathbf{g}_{k}\mathbf{q}_{k}+\mathbf{g}_{k}\mathbf{h}^H_{b,k}+\mathbf{h}_{b,k}\mathbf{q}_{k}+\mathbf{h}_{b,k}\mathbf{h}^H_{b,k}|^2\right]\\
& \quad =P_k\mathbb{E}\left[\mathbf{g}_{k}\mathbf{q}_{k}\mathbf{q}^H_{k}\mathbf{g}^H_{k}+ \mathbf{g}_{k}\mathbf{q}_{k}\mathbf{h}_{b,k}\mathbf{g}^H_{k} +\mathbf{g}_{k}\mathbf{q}_{k}\mathbf{q}^H_{k}\mathbf{h}^H_{b,k} 
%\right. \\ & \quad \qquad \qquad \left. 
+\mathbf{g}_{k}\mathbf{q}_{k}\mathbf{h}_{b,k}\mathbf{h}^H_{b,k}+ \mathbf{g}_{k}\mathbf{h}^H_{b,k}\mathbf{q}^H_{k}\mathbf{g}^H_{k}+ \mathbf{g}_{k}\mathbf{h}^H_{b,k}\mathbf{h}_{b,k}\mathbf{g}^H_{k} +\mathbf{g}_{k}\mathbf{h}^H_{b,k}\mathbf{q}^H_{k}\mathbf{h}^H_{b,k} \right. \\ 
& \qquad\left.
 + \mathbf{g}_{k}\mathbf{h}^H_{b,k}\mathbf{h}_{b,k}\mathbf{h}^H_{b,k}+ \mathbf{h}_{b,k}\mathbf{q}_{k}\mathbf{q}^H_{k} \mathbf{g}^H_{k}+\mathbf{h}_{b,k}\mathbf{q}_{k}\mathbf{h}_{b,k} \mathbf{g}^H_{k}+\mathbf{h}_{b,k}\mathbf{q}_{k}\mathbf{q}^H_{k}\mathbf{h}^H_{b,k} +\mathbf{h}_{b,k}\mathbf{q}_{k}\mathbf{h}_{b,k}\mathbf{h}^H_{b,k} \right. \\ 
 & \qquad \left. +\mathbf{h}_{b,k}\mathbf{h}^H_{b,k}\mathbf{q}^H_{k}\mathbf{g}^H_{k} +\mathbf{h}_{b,k}\mathbf{h}^H_{b,k}\mathbf{h}_{b,k}\mathbf{g}^H_{k} +\mathbf{h}_{b,k}\mathbf{h}^H_{b,k}\mathbf{q}^H_{k}\mathbf{h}^H_{b,k} +\mathbf{h}_{b,k}\mathbf{h}^H_{b,k}\mathbf{h}_{b,k}\mathbf{h}^H_{b,k} \right]\\
 & \quad \overset{(c)}{=}P_k\left(\mathbb{E}\left|\mathbf{g}_{k}\mathbf{q}_{k}\right|^2+ 2\mathbb{E}\left[
 {\rm{Re}}\left\{\mathbf{g}_{k}\mathbf{q}_{k}\mathbf{h}_{b,k}\mathbf{h}^H_{b,k}\right\}\right] +\mathbb{E}\left|\mathbf{g}_{k}\mathbf{h}^H_{b,k}\right|^2 
% \right. \\ & \quad \qquad \qquad\left.
 +2\mathbb{E}\left[
 {\rm{Re}}\left\{\mathbf{g}_{k}\mathbf{h}^H_{b,k}\mathbf{q}^H_{k}\mathbf{h}^H_{b,k}\right\}\right]+ \mathbb{E}\left|\mathbf{h}_{b,k}\mathbf{q}_{k}\right|^2 +\mathbb{E}\left|\mathbf{h}_{b,k}\mathbf{h}^H_{b,k}\right|^2 \right)
\end{aligned}
\end{equation}
\hrulefill
\end{figure*}
where $(c)$ uses the fact that $\mathbb{E}\left[\mathbf{g}_{k}\mathbf{q}_{k}\mathbf{h}_{b,k}\mathbf{g}^H_{k} \right]=0$, $ \mathbb{E}\left[\mathbf{g}_{k}\mathbf{q}_{k}\mathbf{q}^H_{k}\mathbf{h}^H_{b,k}\right]=0$, $ \mathbb{E}\left[\mathbf{g}_{k}\mathbf{h}^H_{b,k}\mathbf{q}^H_{k}\mathbf{g}^H_{k}\right]=0$, and
$\mathbb{E}\left[\mathbf{h}_{b,k}\mathbf{q}_{k}\mathbf{q}^H_{k}\mathbf{g}_{k}\right]=0$ since each element of $\mathbf{h}_{b,k}$ is independent and zero mean, and $\mathbb{E}\left[\mathbf{g}_{k}\mathbf{h}^H_{b,k}\mathbf{h}_{b,k}\mathbf{h}^H_{b,k}\right]=0$, $\mathbb{E}\left[\mathbf{h}_{b,k}\mathbf{q}_{k}\mathbf{h}_{b,k}\mathbf{h}^H_{b,k}\right]=0$, $\mathbb{E}\left[\mathbf{h}_{b,k}\mathbf{h}^H_{b,k}\mathbf{h}_{b,k}\mathbf{g}^H_{k}\right]=0$, and $\mathbb{E}\left[\mathbf{h}_{b,k}\mathbf{h}^H_{b,k}\mathbf{q}^H_{k}\mathbf{h}^H_{b,k}\right]=0$ since $\mathbb{E}\left[s^k\right]=0$ when $k$ is odd and $s$ is normally distributed with zero mean~\cite{winkelbauer2014moments}.

The expectation of each part in~(\ref{e38}) is calculated as:
\subsubsection{$\mathbb{E}\left|\mathbf{g}_{k}\mathbf{q}_{k}\right|^2$} can be represented as (\ref{e39}).
\begin{figure*}[ht]
\begin{equation}\label{e39}
\begin{aligned}
&\mathbb{E}|\mathbf{g}_{k}\mathbf{q}_{k}|^2=\frac{\alpha^2_{r,b}\alpha^2_{k,r}}{(1+\kappa_{k,r})^2}\mathbb{E}\left\{\left| \sum_{u=1}^{2}\sum_{v=1}^{2}\mathbf{g}^{u}_{k}\mathbf{q}^{v}_{k}   \right|^2\right\}
%\\ & \qquad 
=\frac{\alpha^2_{r,b}\alpha^2_{k,r}}{(1+\kappa_{k,r})^2}\Biggl(\mathbb{E}\left\{ \sum_{u=1}^{2}\sum_{v=1}^{2}\left|\mathbf{g}^{u}_{k}\mathbf{q}^{v}_{k}   \right|^2\right\}+2\mathbb{E}\left\{\rm{Re}\left\{ \mathbf{g}^{1}_{k}\mathbf{q}^{1}_{k}(\mathbf{q}^{2}_{k})^H(\mathbf{g}^{1}_{k})^H \right\} \right\} \\ & \qquad \left. +2\mathbb{E}\left\{\rm{Re}\left\{ \mathbf{g}^{1}_{k}\mathbf{q}^{1}_{k}(\mathbf{q}^{1}_{k})^H(\mathbf{g}^{2}_{k})^H \right\} \right\} +2\mathbb{E}\left\{\rm{Re}\left\{ \mathbf{g}^{1}_{k}\mathbf{q}^{1}_{k}(\mathbf{q}^{2}_{k})^H(\mathbf{g}^{2}_{k})^H \right\} \right\} +2\mathbb{E}\left\{\rm{Re}\left\{ \mathbf{g}^{1}_{k}\mathbf{q}^{2}_{k}(\mathbf{q}^{1}_{k})^H(\mathbf{g}^{2}_{k})^H \right\} \right\} \right. \\ & \qquad +2\mathbb{E}\left\{\rm{Re}\left\{ \mathbf{g}^{1}_{k}\mathbf{q}^{2}_{k}(\mathbf{q}^{2}_{k})^H(\mathbf{g}^{2}_{k})^H \right\} \right\} +2\mathbb{E}\left\{\rm{Re}\left\{ \mathbf{g}^{2}_{k}\mathbf{q}^{1}_{k}(\mathbf{q}^{2}_{k})^H(\mathbf{g}^{2}_{k})^H \right\} \right\}
 \Biggr).
\end{aligned}
\end{equation}
\hrulefill
\end{figure*}
First, we can calculate $\mathbb{E}\left| \mathbf{g}^{u}_{k}\mathbf{q}^{v}_{k} \right|^2$, $1\leq u,v \leq 2$. When $u=1$, we have
\begin{eqnarray}
\mathbb{E}\left| \mathbf{g}^{1}_{k}\mathbf{q}^{1}_{k}   \right|^2
&=& \mathbb{E}\left\{\kappa^2_{k,r}\left| \overline{\mathbf{h}}_{r,k}\mathbf{\Phi}^*\overline{\mathbf{H}}_{b,r} \overline{\mathbf{H}}^H_{b,r}\widetilde{\mathbf{\Phi}}\overline{\mathbf{h}}^H_{r,k}  \right|^2\right\} \nonumber \\
&=&\kappa^2_{k,r}M^2\left|f_k(\mathbf{\Phi}^*)\right|^2 \left|f^*_k(\mathbf{\Phi}^*)\right|^2 \label{e40} \\
\mathbb{E}\left| \mathbf{g}^{1}_{k}\mathbf{q}^{2}_{k} \right|^2
&=&\mathbb{E}\left\{\kappa_{k,r}\left| \overline{\mathbf{h}}_{r,k}\mathbf{\Phi}^*\overline{\mathbf{H}}_{b,r} \overline{\mathbf{H}}^H_{b,r}\widetilde{\mathbf{\Phi}}\widetilde{\mathbf{h}}^H_{r,k}  \right|^2\right\} \nonumber \\ &=& \kappa_{k,r}M^2 N\left|f_k(\mathbf{\Phi}^*)\right|^2. \label{e41}
\end{eqnarray}
\begin{comment}
$\mathbb{E}\left\{\left| \mathbf{g}^{1}_{k}\mathbf{q}^{1}_{k}   \right|^2\right\}=\mathbb{E}\left\{\kappa_{k,r}^2\left| \overline{\mathbf{h}}_{r,k}\mathbf{\Phi}^*\overline{\mathbf{H}}_{b,r} \overline{\mathbf{H}}^H_{b,r}\widetilde{\mathbf{\Phi}}\overline{\mathbf{h}}^H_{r,k}  \right|^2\right\}=\kappa_{k,r}^2 M^2\left|f_k(\mathbf{\Phi}^*)\right|^2 \left|f^*_k(\mathbf{\Phi}^*)\right|^2$,\\
$\mathbb{E}\left\{\left| \mathbf{g}^{1}_{k}\mathbf{q}^{2}_{k}   \right|^2\right\}=\mathbb{E}\left\{\kappa_{k,r}\left| \overline{\mathbf{h}}_{r,k}\mathbf{\Phi}^*\overline{\mathbf{H}}_{b,r} \overline{\mathbf{H}}^H_{b,r}\widetilde{\mathbf{\Phi}}\widetilde{\mathbf{h}}^H_{r,k}  \right|^2\right\}=\kappa_{k,r}M^2 N\left|f_k(\mathbf{\Phi}^*)\right|^2$,\\
\end{comment}
Similarly, when $u=2$, we have
\begin{eqnarray}
\mathbb{E}\left| \mathbf{g}^{2}_{k}\mathbf{q}^{1}_{k}   \right|^2
&=&\mathbb{E}\left\{\kappa_{k,r}\left| \widetilde{\mathbf{h}}_{r,k}\mathbf{\Phi}^*\overline{\mathbf{H}}_{b,r} \overline{\mathbf{H}}^H_{b,r}\widetilde{\mathbf{\Phi}}\overline{\mathbf{h}}^H_{r,k}  \right|^2\right\} \nonumber \\ &=& \kappa_{k,r}M^2N \left|f^*_k(\mathbf{\Phi}^*)\right|^2, \label{e42} \\
\mathbb{E}\left| \mathbf{g}^{2}_{k}\mathbf{q}^{2}_{k}   \right|^2
&=&\mathbb{E}\left\{\left| \widetilde{\mathbf{h}}_{r,k}\mathbf{\Phi}^*\overline{\mathbf{H}}_{b,r} \overline{\mathbf{H}}^H_{b,r}\widetilde{\mathbf{\Phi}}\widetilde{\mathbf{h}}^H_{r,k}  \right|^2\right\} \label{e43} \\
&=& \mathbb{E}\left\{M^2\left| \widetilde{\mathbf{h}}_{r,k}\mathbf{t}_1(\mathbf{\Phi}^*)\widetilde{\mathbf{h}}^H_{r,k}  \right|^2\right\}
=M^2 \Psi_1(\mathbf{\Phi}^*). \nonumber
\end{eqnarray}
\begin{comment}
$\mathbb{E}\left\{\left| \mathbf{g}^{2}_{k}\mathbf{q}^{1}_{k}   \right|^2\right\}=\mathbb{E}\left\{\kappa_{k,r}\left| \widetilde{\mathbf{h}}_{r,k}\mathbf{\Phi}^*\overline{\mathbf{H}}_{b,r} \overline{\mathbf{H}}^H_{b,r}\widetilde{\mathbf{\Phi}}\overline{\mathbf{h}}^H_{r,k}  \right|^2\right\}=\kappa_{k,r}M^2N \left|f^*_k(\mathbf{\Phi}^*)\right|^2$,\\
$\mathbb{E}\left\{\left| \mathbf{g}^{2}_{k}\mathbf{q}^{2}_{k}   \right|^2\right\}=\mathbb{E}\left\{\left| \widetilde{\mathbf{h}}_{r,k}\mathbf{\Phi}^*\overline{\mathbf{H}}_{b,r} \overline{\mathbf{H}}^H_{b,r}\widetilde{\mathbf{\Phi}}\widetilde{\mathbf{h}}^H_{r,k}  \right|^2\right\}=M^2 \Psi_1(\mathbf{\Phi}^*)$,\\
\end{comment}
The remaining terms in (\ref{e39}) can be eliminated since
\begin{equation*}
\begin{aligned}
\mathbb{E}\left\{\rm{Re}\left\{ \mathbf{g}^{1}_{k}\mathbf{q}^{1}_{k}(\mathbf{q}^{2}_{k})^H(\mathbf{g}^{1}_{k})^H \right\} \right\} &=0 \quad\text{since} \; \widetilde{\mathbf{h}}_{r,k} \; \text{is zero mean}\\
\mathbb{E}\left\{\rm{Re}\left\{ \mathbf{g}^{1}_{k}\mathbf{q}^{1}_{k}(\mathbf{q}^{1}_{k})^H(\mathbf{g}^{2}_{k})^H \right\} \right\} &=0 \quad \text{same reason as above} \\
\mathbb{E}\left\{\rm{Re}\left\{ \mathbf{g}^{1}_{k}\mathbf{q}^{1}_{k}(\mathbf{q}^{2}_{k})^H(\mathbf{g}^{2}_{k})^H \right\} \right\} &=0 \quad \text{due to (\ref{e29})}\\
\mathbb{E}\left\{\rm{Re}\left\{ \mathbf{g}^{1}_{k}\mathbf{q}^{2}_{k}(\mathbf{q}^{1}_{k})^H(\mathbf{g}^{2}_{k})^H \right\} \right\} &=0 \quad \text{same reason as above}\\
\mathbb{E}\left\{\rm{Re}\left\{ \mathbf{g}^{1}_{k}\mathbf{q}^{2}_{k}(\mathbf{q}^{2}_{k})^H(\mathbf{g}^{2}_{k})^H \right\} \right\} &=0 \quad \text{since} \; \mathbb{E}\left[s^k\right]=0 \; \text{for} \\ & \qquad \quad \text{odd Gaussian moments}\\
\mathbb{E}\left\{\rm{Re}\left\{ \mathbf{g}^{2}_{k}\mathbf{q}^{1}_{k}(\mathbf{q}^{2}_{k})^H(\mathbf{g}^{2}_{k})^H \right\} \right\} &=0 \quad \text{same reason as above}.
\end{aligned}
\end{equation*}

\subsubsection{$ \mathbb{E}\left[
{\rm{Re}}\left\{\mathbf{g}_{k}\mathbf{q}_{k}\mathbf{h}_{b,k}\mathbf{h}^H_{b,k}\right\}\right]$} Since $\mathbf{h}_{b,k}$ is independent of $\mathbf{g}_{k}$ and $\mathbf{q}_{k}$, we have 
 \begin{equation}
 \mathbb{E}\left[
 {\rm{Re}}\left\{\mathbf{g}_{k}\mathbf{q}_{k}\mathbf{h}_{b,k}\mathbf{h}^H_{b,k}\right\}\right]=\alpha_{k,b}M \mathbb{E}\left[
 {\rm{Re}}\left\{\mathbf{g}_{k}\mathbf{q}_{k}\right\}\right] .
 \end{equation}
Considering $\widetilde{\mathbf{h}}_{r,k}$, we can ignore the zero expectation part in $\mathbb{E}\left[
 \mathbf{g}_{k}\mathbf{q}_{k}\right]$ and obtain the following 2 results
\begin{equation}\label{e44}
\begin{aligned}
&\mathbb{E}\left[
 {\rm{Re}}\left\{\kappa_{k,r}\overline{\mathbf{h}}_{r,k}\mathbf{\Phi}^*\overline{\mathbf{H}}_{b,r} \overline{\mathbf{H}}^H_{b,r}\widetilde{\mathbf{\Phi}}\overline{\mathbf{h}}^H_{r,k}  \right\}\right]\\
 & \quad =\mathbb{E}\left[
 M{\rm{Re}}\left\{\kappa_{k,r}\overline{\mathbf{h}}_{r,k}\mathbf{\Phi}^*\mathbf{a}_N\mathbf{a}^H_N \widetilde{\mathbf{\Phi}}\overline{\mathbf{h}}^H_{r,k}  \right\}\right]\\
 & \quad =\kappa_{k,r} M {\rm{Re}}\left\{f_k(\mathbf{\Phi}^*)f^*_k(\mathbf{\Phi}^*) \right\},\\
&\mathbb{E}\left[
 {\rm{Re}}\left\{\widetilde{\mathbf{h}}_{r,k}\mathbf{\Phi}^*\overline{\mathbf{H}}_{b,r} \overline{\mathbf{H}}^H_{b,r}\widetilde{\mathbf{\Phi}}\widetilde{\mathbf{h}}^H_{r,k}  \right\}\right]\\
 & \quad =\mathbb{E}\left[
 M{\rm{Re}}\left\{{\rm{Tr}}\left(\mathbf{\Phi}^*\mathbf{a}_N\mathbf{a}^H_N \widetilde{\mathbf{\Phi}} \right) \right\}\right]
% \\ & \qquad
 =M{\rm{Re}}\left\{{\rm{Tr}}\left(\mathbf{t}_1(\mathbf{\Phi}^*)\right)\right\}.
\end{aligned}
\end{equation}
\begin{comment}
 $\mathbb{E}\left[
 {\rm{Re}}\left\{\kappa_{k,r}\overline{\mathbf{h}}_{r,k}\mathbf{\Phi}^*\overline{\mathbf{H}}_{b,r} \overline{\mathbf{H}}^H_{b,r}\widetilde{\mathbf{\Phi}}\overline{\mathbf{h}}^H_{r,k}  \right\}\right]=\kappa_{k,r} M {\rm{Re}}\left\{f_k(\mathbf{\Phi}^*)f^*_k(\mathbf{\Phi}^*) \right\}$,\\
 $
 \mathbb{E}\left[
 {\rm{Re}}\left\{\widetilde{\mathbf{h}}_{r,k}\mathbf{\Phi}^*\overline{\mathbf{H}}_{b,r} \overline{\mathbf{H}}^H_{b,r}\widetilde{\mathbf{\Phi}}\widetilde{\mathbf{h}}^H_{r,k}  \right\}\right]=
M{\rm{Re}}\left\{{\rm{Tr}}\left(\mathbf{t}_1(\mathbf{\Phi}^*)\right)\right\}$,\\
\end{comment}
Combining the above together, we have
\begin{equation}\label{e46}
\begin{aligned}
&\mathbb{E}\left[
 {\rm{Re}}\left\{\mathbf{g}_{k}\mathbf{q}_{k}\mathbf{h}^H_{k,b_1}\mathbf{h}_{k,b_1}\right\}\right]\\&  \qquad= \frac{M\alpha_{k,b}\alpha_{k,r}\alpha_{r,b}}{(1+\kappa_{k,r})}\left(\kappa_{k,r} M {\rm{Re}}\left\{f_k(\mathbf{\Phi}^*)f^*_k(\mathbf{\Phi}^*) \right\} \right. \\ 
 & \qquad\qquad \left.+M{\rm{Re}}\left\{{\rm{Tr}}\left(\mathbf{t}_1(\mathbf{\Phi}^*)\right)\right\}\right)
 \\&  \qquad= \frac{M^2\alpha_{k,b}\alpha_{k,r}\alpha_{r,b}}{(1+\kappa_{k,r})}\left(\kappa_{k,r} {\rm{Re}}\left\{f_k(\mathbf{\Phi}^*)f^*_k(\mathbf{\Phi}^*) \right\} \right. \\ 
 & \qquad\qquad \left. +{\rm{Re}}\left\{{\rm{Tr}}\left(\mathbf{t}_1(\mathbf{\Phi}^*)\right)\right\}\right).
\end{aligned}
\end{equation}

\subsubsection{$ \mathbb{E}\left|\mathbf{g}_{k}\mathbf{h}^H_{b,k}\right|^2$ and $\mathbb{E}\left|\mathbf{h}_{b,k}\mathbf{q}_{k}\right|^2$}
First, we can directly calculate $\mathbb{E}|\mathbf{g}_{k}\mathbf{h}^H_{b,k}|^2$ as follows:
\begin{equation}\label{e47}
\begin{aligned}
\mathbb{E}|\mathbf{g}_{k}\mathbf{h}^H_{b,k}|^2 &=\alpha_{k,b}\mathbb{E}\left[\mathbf{g}_{k}\mathbf{g}^H_{k}\right]=
\frac{\alpha_{k,b}\alpha_{r,b}\alpha_{k,r}}{(1+\kappa_{k,r})}\sum_{u=1}^{2}\mathbb{E} \Vert\mathbf{g}^u_{k}\Vert^2\\
&\overset{(b)}{=}\frac{\alpha_{k,b}\alpha_{r,b}\alpha_{k,r}\left(\kappa_{k,r} M\left|f_k(\mathbf{\Phi}^*)\right|^2+MN\right)}{(1+\kappa_{k,r})}.
\end{aligned}
\end{equation}
Similarly, $\mathbb{E}\left|\mathbf{h}_{b,k}\mathbf{q}_{k}\right|^2$ can be calculated as follows:
\begin{equation}\label{e48}
\begin{aligned}
\mathbb{E}\left|\mathbf{h}_{b,k}\mathbf{q}_{k}\right|^2&=\alpha_{k,b}\mathbb{E}\left[\mathbf{q}^H_{k}\mathbf{q}_{k}\right]
%\\ &
=\frac{\alpha_{k,b}\alpha_{r,b}\alpha_{k,r}}{(1+\kappa_{k,r})}\sum_{v=1}^{2} \mathbb{E}\Vert\mathbf{q}^u_{k}\Vert^2\\
&\overset{(b)}{=}\frac{\alpha_{k,b}\alpha_{r,b}\alpha_{k,r}\left(\kappa_{k,r} M\left|f^*_k(\mathbf{\Phi}^*)\right|^2+MN\right)}{(1+\kappa_{k,r})}.
\end{aligned}
\end{equation}

\subsubsection{$\mathbb{E}\left[
 {\rm{Re}}\left\{\mathbf{g}_{k}\mathbf{h}^H_{b,k}\mathbf{q}^H_{k}\mathbf{h}^H_{b,k}\right\}\right] $} Based on the following analysis, we have $\mathbb{E}\left[
 {\rm{Re}}\left\{\mathbf{g}_{k}\mathbf{h}^H_{b,k}\mathbf{q}^H_{k}\mathbf{h}^H_{b,k}\right\}\right] =0$.
\begin{proof}
According to (\ref{e23}) and (\ref{e24}), $\mathbb{E}\left[
\mathbf{g}_{k}\mathbf{h}^H_{b,k}\mathbf{q}^H_{k}\mathbf{h}^H_{b,k}\right]=\mathbb{E}\left[\mathbf{h}_{r,k}\mathbf{\Phi}^* \mathbf{H}_{b,r}\mathbf{h}^H_{b,k}\mathbf{h}_{r,k}(\widetilde{\mathbf{\Phi}})^H \mathbf{H}_{b,r}\mathbf{h}^H_{b,k}\right]$ can be divided into 4 terms. Given the independence and zero mean of each element in $\mathbf{h}^H_{b,k}$ and $\widetilde{\mathbf{h}}_{r,k}$, two of the terms have an expectation of 0. Then we calculate the expectations of the remaining two terms. Since $\mathbf{h}^H_{b,k}$ follows the Rayleigh fading distribution, using~(\ref{e29}), we have that
\begin{equation}
\mathbb{E}\left[{\rm{Re}}\left\{\overline{\mathbf{h}}_{r,k}\mathbf{\Phi}^* \overline{\mathbf{H}}_{b,r}\mathbf{h}^H_{b,k}\overline{\mathbf{h}}_{r,k}(\widetilde{\mathbf{\Phi}})^H \overline{\mathbf{H}}_{b,r}\mathbf{h}^H_{b,k}\right\}\right]=0 .
\end{equation}
For the other term, we set $\mathbf{w}_1(\mathbf{\Phi} )=\widetilde{\mathbf{h}}_{r,k}\mathbf{\Phi}^*=\widetilde{\mathbf{h}}_{r,k}\mathbf{J} \mathbf{\Phi}$, $\mathbf{w}_2=\overline{\mathbf{H}}_{b,r}\mathbf{h}^H_{b,k}$, and $\mathbf{w}_3(\mathbf{\Phi} )=\widetilde{\mathbf{h}}_{r,k}(\widetilde{\mathbf{\Phi}})^H=\widetilde{\mathbf{h}}_{r,k}\mathbf{\Phi} \mathbf{J}^T$. According to the mapping $\mathbf{J}: n\to [n]$ and assume the mapping rule: $\mathbf{J}^T: n\to [n]^*$, we have
\begin{equation}
\begin{aligned}
\mathbb{E}\left[\widetilde{\mathbf{h}}_{r,k}\mathbf{\Phi}^* \overline{\mathbf{H}}_{b,r}\mathbf{h}^H_{b,k}\widetilde{\mathbf{h}}_{r,k}(\widetilde{\mathbf{\Phi}})^H \overline{\mathbf{H}}_{b,r}\mathbf{h}^H_{b,k}\right]\\=\mathbb{E}\left[\mathbf{w}_1(\mathbf{\Phi} )\mathbf{w}_2\mathbf{w}_3(\mathbf{\Phi} )\mathbf{w}_2\right]
\end{aligned}
\end{equation}
with (\ref{e49}), which is denoted as follows:
\begin{equation}\label{e49}
\begin{aligned}
\mathbf{w}_1(\mathbf{\Phi} )\mathbf{w}_2&=\sum_{j=1}^{N}\left(\left[\widetilde{\mathbf{h}}_{r,k}\right]_{[j]}\theta_j \sum_{i=1}^{M}\left[\overline{\mathbf{H}}_{b,r}\right]_{j,i} [\mathbf{h}^H_{b,k}]_i\right), \\
\mathbf{w}_3(\mathbf{\Phi} )\mathbf{w}_2&=\sum_{j=1}^{N}\left(\left[\widetilde{\mathbf{h}}_{r,k}\right]_{j}\theta_j \sum_{i=1}^{M}\left[\overline{\mathbf{H}}_{b,r}\right]_{[j]^*,i} [\mathbf{h}^H_{b,k}]_i\right).
\end{aligned}
\end{equation}
Substituting (\ref{e49}), we have (\ref{e51}).
\begin{figure*}
\begin{equation}\label{e51}
\begin{aligned}
&\mathbb{E}\left[\mathbf{w}_1(\mathbf{\Phi} )\mathbf{w}_2\mathbf{w}_3(\mathbf{\Phi} )\mathbf{w}_2\right]=\mathbb{E}\left[\sum_{p=1}^{N}\sum_{q=1}^{N}\left( \left[\widetilde{\mathbf{h}}_{r,k}\right]_{[p]}\theta_p\left[\widetilde{\mathbf{h}}_{r,k}\right]_{q}\theta_q \left(\sum_{i=1}^{M}\left[\overline{\mathbf{H}}^H_{r,b}\right]_{p,i} [\mathbf{h}^H_{b,k}]_i \right)\left(\sum_{i=1}^{M}\left[\overline{\mathbf{H}}^H_{r,b}\right]_{[q]^*,i} [\mathbf{h}^H_{b,k}]_i \right) \right) \right].
\end{aligned}
\end{equation}
\hrulefill
\end{figure*}
Given the independence and zero mean of $[\mathbf{h}^H_{b,k}]_i$ and $\left[\widetilde{\mathbf{h}}_{r,k}\right]_i$, the expectation of items with $\left[\widetilde{\mathbf{h}}_{r,k}\right]_{[p]}\left[\widetilde{\mathbf{h}}_{r,k}\right]_{q}$ in (\ref{e47}) is 0 if $[p]\neq q$. Since $\left[\widetilde{\mathbf{h}}_{r,k}\right]_i \sim \mathcal{CN}(0,1)$, we set $\left[\widetilde{\mathbf{h}}_{r,k}\right]_i =a+jb$, where $a,b \sim \mathcal{N}(0,\frac{1}{2})$. Then $\mathbb{E}\left(\left[\widetilde{\mathbf{h}}_{r,k}\right]_i\right)^2=\mathbb{E}\left[ (a+jb)^2\right]=\mathbb{E}\left[ a^2+2jab-b^2\right]=0$, and we find that the expectation of items with $[p]=q$ in (\ref{e51}) is 0. Finally, we can prove $\mathbb{E}\left[
 {\rm{Re}}\left\{\mathbf{g}_{k}\mathbf{h}_{k,b_1}\mathbf{q}^H_{k}\mathbf{h}_{k,b_1}\right\}\right] =0$. $\hfill\blacksquare$
\end{proof}

\subsubsection{$\mathbb{E}\left|\mathbf{h}_{b,k}\mathbf{h}^H_{b,k}\right|^2 $}According to \cite[Appendix A]{9355404}, we have
\begin{equation}\label{e52}
\begin{aligned}
\mathbb{E}\left|\mathbf{h}_{b,k}\mathbf{h}^H_{b,k}\right|^2 =\alpha^2_{k,b}(M^2+M).
\end{aligned}
\end{equation}
Combining the above equations and after some simplifying steps, we have completed the proof of 
Theorem~\ref{t1}.

\bibliographystyle{IEEEtran}
\bibliography{crack}

\end{document}